\documentclass[11pt]{article}
\usepackage{graphicx,amsmath,amssymb,amsthm,enumerate}
\usepackage{hyperref}
\usepackage{microtype}
\usepackage{enumitem}
\usepackage{wrapfig}
\usepackage{complexity}
\usepackage{xcolor,tcolorbox}
\usepackage{lineno}
\usepackage[ruled,vlined]{algorithm2e}
\usepackage{bm}
\usepackage{capt-of}
\usepackage{float}
\usepackage{anyfontsize}
\usepackage{tikz}
\usetikzlibrary{positioning}
\usetikzlibrary{automata,arrows}
\usetikzlibrary{calc}
\usetikzlibrary{shapes}
\usetikzlibrary{decorations}
\usetikzlibrary{arrows,positioning, quotes, arrows.meta}
\usetikzlibrary{patterns,decorations.pathreplacing,calligraphy}
 
\newtheorem{lem}{Lemma} 
\newcommand {\BL} {\begin{lem}} 

\newtheorem{obs}{Observation} 
\newcommand {\BO} {\begin{obs}} 
\newcommand {\EO} {\end{obs}} 

\newtheorem{defn}{Definition} 
\newcommand {\BDE} {\begin{defn}} 
\newcommand {\EDE} {\end{defn}} 

\newtheorem{cor}  {Corollary} 
\newcommand {\BCR} {\begin{cor}}
\newcommand {\ECR} {\end{cor}}

\newtheorem{thm} {Theorem} 
\newcommand {\BT} {\begin{thm}}
\newcommand {\ET} {\end{thm}}

\newtheorem{fact}  {Fact} 
\newcommand {\BF} {\begin{fact}} 
\newcommand {\EF} {\end{fact}} 

\newtheorem{convention}  {Convention} 
\newcommand {\BC} {\begin{convention}} 
\newcommand {\EC} {\end{convention}} 

\newtheorem{hyp}{Hypothesis}

\newtheorem{clm}  {Claim} 

\newcommand {\BCL} {\begin{clm}} 
\newcommand {\ECL} {\end{clm}} 

\newtheorem{problem}{Problem}

\title{
Simple Reductions from Formula-SAT to Pattern Matching on Labeled Graphs and Subtree Isomorphism
}

\author{
Daniel Gibney\thanks{Dept. of CS, University of Central Florida, Orlando, USA. \textbf{e-mail: } daniel.j.gibney@gmail.com}
\and
Gary Hoppenworth\thanks{Dept. of CS, University of Central Florida, Orlando, USA. \textbf{e-mail: } gary.hoppenworth@gmail.com}
\and 
Sharma V. Thankachan\thanks{Dept. of CS, University of Central Florida, Orlando, USA. \textbf{e-mail: } sharma.thankachan@ucf.edu}
}

\date{}

\begin{document}
\maketitle

\begin{abstract}
The CNF formula satisfiability problem (CNF-SAT) has been reduced to many fundamental problems in $\P$ to prove tight lower bounds under the Strong Exponential Time Hypothesis (SETH). Recently, the works of Abboud, Hansen, Vassilevska W. and Williams (STOC’16), and later, Abboud and Bringmann (ICALP’18) have proposed basing lower bounds on the hardness of general boolean formula satisfiability (Formula-SAT). Reductions from Formula-SAT have two advantages over the usual reductions from CNF-SAT: (1) conjectures on the hardness of Formula-SAT are arguably much more plausible than those of CNF-SAT, and (2) these reductions give consequences even for logarithmic improvements in a problem’s upper bounds. 

Here we give tight reductions from Formula-SAT to two more problems: pattern matching on labeled graphs (PMLG) and subtree isomorphism. Previous reductions from Formula-SAT were to sequence alignment problems such as Edit Distance, LCS, and Frechet Distance and required some technical work. This paper uses ideas similar to those used previously, but in a decidedly simpler setting, helping to illustrate the most salient features of the underlying techniques.
\end{abstract}

\newpage
\section{Introduction and Related Work}
The Strong Exponential Time Hypothesis (SETH) has proven to be a powerful tool in establishing conditional lower bounds for many problems with known polynomial-time solutions. However, recent work by Abboud, Hansen, Vassilevska W., and Williams~\cite{DBLP:conf/stoc/AbboudHWW16}, as well as Abboud and Bringmann~\cite{DBLP:conf/icalp/AbboudB18} has sought to use the hardness of general Formula-SAT problems as the basis for fine-grained conditional lower bounds, rather than CNF-SAT and SETH. Since general Formula-SAT contains within it all CNF-SAT instances, Formula-SAT is at least as hard as CNF-SAT. Additionally, when basing conditional lower bounds on Formula-SAT rather than CNF-SAT, the same algorithmic breakthroughs that previously would have violated SETH, now have far more remarkable consequences (see Section \ref{sec:contribution} for examples). This makes it plausible that conjectures based on the hardness of Formula-SAT are more likely to hold than those based on the hardness of CNF-SAT.

Aside from a plausible increase in the robustness of the conjectures, using Formula-SAT as a starting point has the advantage of allowing for tighter hardness results. Previous lower bounds based on SETH have been effective in establishing results of the form: an algorithm running in time $\mathcal{O}(n^{c-\varepsilon})$ for some $\varepsilon > 0$, where the best-known solution has time complexity $\widetilde{\mathcal{O}}(n^{c})$ would violate SETH. Despite this success, SETH has proven less effective at establishing tighter fine-grained hardness results regarding how many logarithmic-factors can be shaved. In fact, the impossibility of proving such a hardness result via fine-grained reductions from CNF-SAT was proven in~\cite{DBLP:conf/icalp/AbboudB18}. Overcoming this by using Formula-SAT as a starting point, in \cite{DBLP:conf/stoc/AbboudHWW16} conditional lower bounds of this form were established for Edit Distance and Longest Common Subsequence (LCS).  In \cite{DBLP:conf/icalp/AbboudB18}, the results on LCS were further extended to show that an $\mathcal{O}(n^2/ \log^{7+\varepsilon} n)$ time solution for LCS would imply major breakthroughs in circuit complexity. As a final example, work in \cite{DBLP:journals/corr/abs-2008-02769} uses reductions from Formula-SAT to analyze which regular expression matching problems can have super-polylog factors shaved from their time complexity, and which cannot.

In this work, we will use Formula-SAT to establish hardness results similar to those listed above, but for two additional fundamental problems, Pattern Matching on Labeled Graphs (PMLG) and Subtree Isomorphism. We describe these problems next.

\bigskip
\noindent
\textbf{Pattern Matching On Labeled Graphs. (PMLG)}
Given an alphabet $\Sigma$, a labeled graph $G$ is a triplet $(V, E, L)$, where $(V, E)$ corresponds to the vertices and edges of a graph, and $L : V \rightarrow \Sigma^+$ is a function that defines a nonempty string (i.e., label) over $\Sigma$ to each vertex in $G$. For any string $S$,  we use $S[..\ell]$ to denote its prefix ending at $\ell$ and $S[\ell..]$ to denote its suffix starting at $\ell$.
We say that a pattern $P$ occurs in $G$ if there is a path $v_1, v_2, \dots, v_m$ in $G$ such that $L(v_1)[\ell ..] \circ L(v_2) \circ \cdots \circ L(v_m)[.. \ell']$ equals $P$ for some $\ell, \ell'$. Given a labeled graph $G$ and a pattern $P$, the PMLG problem is to decide if there exists an occurrence of $P$ in $G$

The PMLG problem began being intensely studied roughly thirty years ago in the context of alignment of strings (equivalent to approximate matching under edits, mismatches, etc.) in \emph{hypertext}. This was initiated by Manber and Wu \cite{manber1992approximate} and underwent several improvements \cite{DBLP:conf/cpm/Akutsu93, DBLP:journals/jal/AmirLL00, DBLP:journals/tcs/Navarro00,DBLP:conf/cpm/ParkK95}. In the case where changes are allowed in the pattern, but not in the graph, the best-known algorithm runs in time $\mathcal{O}(|V|+ |E||P|)$, matching the time complexity of the dynamic programming solution of the exact problem, and is by Rautiainen and Marschall~\cite{rautiainen2017aligning}. In the case where changes are allowed in the graph as well, the problem is NP-complete~\cite{DBLP:journals/jal/AmirLL00}, even for binary alphabet~\cite{DBLP:conf/recomb/JainZGA19}. The work by Equi \textit{et al.} in \cite{DBLP:conf/icalp/EquiGMT19} established the SETH based lower bounds for exact matching. 

\bigskip
\noindent
\textbf{Subtree Isomorphism.}
Given two trees $T_1$ and $T_2$, is $T_1$ contained in $T_2$? This problem has been the subject of extensive study~\cite{DBLP:journals/jal/Chung87,DBLP:conf/caap/Lingas83,DBLP:journals/ipl/LingasK89,DBLP:journals/siamcomp/Reyner77,DBLP:journals/jal/ShamirT99,DBLP:journals/siamcomp/VermaR89}, much of this research dating back several decades. For general trees, both with at most $n$ vertices, the currently best known solution has a time bound that is $\mathcal{O}(n^{\omega})$, where $\omega$ is the exponent on fast-matrix multiplication~\cite{DBLP:journals/jal/ShamirT99}; for rooted, constant maximum degree trees it is $\mathcal{O}(n^2/\log n)$~\cite{DBLP:conf/caap/Lingas83}; and, for ordered trees it is $\mathcal{O}(n \log n)$~\cite{DBLP:journals/siamcomp/ColeH03}. Here we will be considering rooted trees with constant maximum degree. In terms of lower bounds, SETH based quadratic lower bounds for this version of the problem have been established in \cite{DBLP:journals/talg/AbboudBHWZ18}, even for binary rooted trees.

\bigskip
\noindent
\textbf{Road Map.} We will first describe the Formula-SAT problem and deMorgan Formulas in more detail. Following this, we will state our results for PMLG and Subtree Isomorphism in terms of its implications for solving Formula-SAT, along with the resulting corollaries. Section \ref{sec:PMLG} provides the reduction from Formula-SAT to PMLG. The reduction to Subtree Isomorphism is in Section \ref{sec:subtree}. Finally, in Section \ref{sec:discussion} we discuss the similar themes and techniques that appear in both of these reductions.

\subsection{Formula-SAT}

\bigskip
\noindent
\textbf{deMorgan Formulas.} For our purposes, we define a deMorgan formula over $n$ Boolean input variables as a rooted binary tree where each leaf node represents an input variable or its negation, and every internal node represents a logical operator from the set $\{\land, \lor\}$. 
Leaf nodes will be called input gates, and internal nodes will be called AND/OR gates. 
For a given bit assignment $x$, we define $F(x)$ as the binary value output at the root of $F$ when the input bits are propagated from the leaves to the root of $F$.
The size of the formula, which we will denote as $s$, is defined as the number of leaves in the tree.

\begin{problem}[Formula-SAT] Given a deMorgan formula $F$ of size $s$ over $n$ inputs, does there exist an input $x \in \{0, 1\}^n$ such that $F(x) = 1$?
\end{problem}
The set of all Formula-SAT instances obviously contains within it all CNF-SAT instances. Unsurprisingly, due to its generality, it appears harder to derive efficient solutions for Formula-SAT. For CNF-SAT there exists ever-improving upper bounds \cite{ DBLP:journals/endm/BruggemannK04,DBLP:conf/stoc/HansenKZZ19,DBLP:journals/dam/MonienS85,DBLP:journals/jacm/PaturiPSZ05, DBLP:conf/aisc/Rodosek96, DBLP:journals/algorithmica/Schoning02}. There also exists upper bounds for more general circuits such as ours, however, these work through restricting some parameter of the circuit, often some combination of the size, depth, and type of gates used within it  (see for example~\cite{DBLP:conf/mfcs/Chen15, impagliazzo2012satisfiability, DBLP:conf/focs/ImpagliazzoPS13, DBLP:journals/eccc/SakaiSTT15,DBLP:journals/cc/SetoT13,DBLP:journals/eccc/Tamaki16}).

\subsection{Our Results}
\label{sec:contribution}
Our reduction will create an instance of PMLG (or Subtree Isomorphism) from a given instance of Formula-SAT. In doing so, we make explicit the roles that the size of the circuit $s$ and the number of inputs $n$ play in determining the size of the resulting instance. 

\begin{thm} \label{thm_pmlg}
A Formula-SAT instance of size $s$ on $n$ inputs can be reduced to an instance of PMLG over a binary alphabet with a graph $G=(V, E)$ and pattern $P$ such that $|P| $ is of size $\mathcal{O}(2^{n/2}\cdot s)$ and $ |E|$ is of size $\mathcal{O}(2 ^{n/2} \cdot s^2 )$ in $\mathcal{O}(|E|)$ time, where $G$ is a DAG with  maximum total degree\footnote{Total degree is in-degree plus out-degree.}  three.
\end{thm}
Similarly, for Subtree Isomorphism we have the following theorem.
\begin{thm} \label{thm_sub_tree}
A Formula-SAT instance of size $s$ on $n$ inputs can be reduced to an instance of Subtree Isomorphism on two binary trees $T_1$ and $T_2$, where the size of $T_1$ is $\mathcal{O}(2^{n/2} \cdot s)$, and the size of $T_2$ is $\mathcal{O}(2^{n/2} \cdot s^2)$ in $\mathcal{O}(|T_2|)$ time.
\end{thm}

Combining Theorems \ref{thm_pmlg} and \ref{thm_sub_tree} with observations made by Abboud \textit{et al.} in \cite{DBLP:conf/stoc/AbboudHWW16} (and restated in Appendix \ref{app:circuit}), we obtain the following `breakthrough' implications of a strongly subquadratic time algorithm for PMLG or Subtree Isomorphism. Proofs are deferred to Appendix \ref{app:circuit}.
\begin{cor}
\label{cor:strict_subquad_cons}
The existence of a strongly subquadratic time algorithm for PMLG (or Subtree Isomorphism) would imply the class $\E^\NP$ (1) does not have non-uniform $2^{o(n)}$-size Boolean formulas and (2) does not have non-uniform $o(n)$-depth circuits of bounded fan-in. It also implies that $\NTIME[2^{\mathcal{O}(n)}]$ is not in non-uniform $\NC$.
\end{cor}
The second corollary gives the consequences of being able to shave arbitrarily many logarithmic factors from the quadratic time complexity.
\begin{cor}
\label{cor:infinite_logshaving_cons}
If PMLG (or Subtree Isomorphism) can be solved in time $\mathcal{O}(\frac{|E||P|}{\log^c |E|})$ or $\mathcal{O}(\frac{|E||P|}{\log^c |P|})$ ( $\mathcal{O}(\frac{|T_1||T_2|}{\log^c |T_1|})$ or $\mathcal{O}(\frac{|T_1||T_2|}{\log^c |T_2|})$ resp.) for all $c = \Theta(1)$, then $\NTIME[2^{O(n)}]$ does not have non-uniform polynomial-size log-depth circuits.
\end{cor}

In fact, we can give a particular constant $c$ for which shaving a $\log^c n$ factor would yield surprising new results in complexity theory. The following log-sensitive lower bounds leave a huge gap from the best known upper bounds; we present these corollaries purely for instructive purposes.

\bigskip
\noindent
\textbf{Hardness of Shaving Log Factors.}
We work under the Word-RAM model and limit the set of constant-time primitive operations to those operations which are robust to change in word size. Specifically, suppose we are given a word size of $w = \Theta(\log n)$ and an operation that can be performed in $\mathcal{O}(1)$ time. We stipulate that we must be able to simulate this operation on words of size $W = \Theta(2^w)$ in time $n^{1 + o(1)}$. This is a reasonable assumption that is satisfied by many constant time operations such as addition, subtraction, multiplication, and division with remainder. See \cite{DBLP:conf/icalp/AbboudB18} for a detailed discussion.

The following hypothesis was suggested by Abboud and Bringmann in \cite{DBLP:conf/icalp/AbboudB18}. It reflects the fact that the best known algorithmic solutions to Formula-SAT\footnote{As observed by Williams in \cite{williams2014algorithms}, for deMorgan formulas of size $n^{3-o(1)}$ there exists a randomized $2^{n-n^{\Omega(1)}}$ time, zero error algorithm which can be obtained by applying results from \cite{DBLP:journals/cc/ChenKKSZ15} and \cite{DBLP:conf/focs/KomargodskiRT13}.} 
fail to provide a time complexity better than the na\"ive solution on formulas of size $s = n ^{3 + \Omega(1)}$.
\begin{hyp}[\cite{DBLP:conf/icalp/AbboudB18}]
\label{hyp1}
There is no algorithm that can solve SAT on deMorgan formulas of size $s = n ^{3 + \Omega(1)}$ in $\mathcal{O}(\frac{2^n}{n^{\varepsilon}})$ time for some $\varepsilon > 0$ in the Word-RAM model.
\end{hyp}
\begin{cor} \label{cor:pmlg}
Hypothesis \ref{hyp1} is false if PMLG (respectively Subtree Isomorphism) can be solved in time $\mathcal{O}\left(\frac{|E||P|}{\log ^ {10 + \varepsilon} |E|}\right)$ or $\mathcal{O}\left(\frac{|E||P|}{\log ^ {10 + \varepsilon} |P|}\right)$, (respectively $\mathcal{O}\left(\frac{|T_1||T_2|}{\log ^ {10 + \varepsilon} |T_1|}\right)$ or $\mathcal{O}\left(\frac{|T_1||T_2|}{\log ^ {10 + \varepsilon} |T_2|}\right)$) for any $\varepsilon > 0$.
\end{cor}
\begin{proof}
We show the proof for PMLG; the proof for Subtree Isomorphism is identical. By Theorem \ref{thm_pmlg}, an  $\mathcal{O}(\frac{|E||P|}{\log ^ {10 + \varepsilon} |E|})$ algorithm for PMLG can be converted to yield an algorithm running in 
$n^{1+o(1)} \cdot \frac{(2^{n/2} \cdot s^2)(2^{n/2}s) }{ \log ^ {10 + \varepsilon} (2^{n / 2} \cdot s^2)} = \mathcal{O}\left(\frac{2^n \cdot s^3}{ n ^ {9 + \varepsilon}}\right)$ time
for Formula-SAT (note the $n^{1 + o(1)}$ factor introduced when moving from a word size of $\Theta(\log n)$ to $\Theta(n)$). If we choose $s = n ^ {3 + \varepsilon/6}$ 
then this yields an algorithm for Formula-SAT of time $\mathcal{O}(\frac{2^n}{n^{\varepsilon/2}})$, and Hypothesis \ref{hyp1} is false.
\end{proof}

Again thanks to results highlighted by Abboud \textit{et al.} in \cite{DBLP:conf/stoc/AbboudHWW16}, we can also say the following about shaving a constant number of logarithmic factors from the quadratic time complexity. The proof is deferred to Appendix \ref{app:circuit}.
\begin{cor} \label{corr:E_NP_log_shaving_constant}
$\E^\NP$ cannot be computed by non-uniform formulas of cubic size if PMLG (respectively Subtree Isomorphism) can be solved in time
$\mathcal{O}\left(\frac{|E||P|}{\log ^ {20 + \varepsilon} |E|}\right)$ or $\mathcal{O}\left(\frac{|E||P|}{\log ^ {20 + \varepsilon} |P|}\right)$ (respectively $\mathcal{O}\left(\frac{|T_1||T_2|}{\log ^ {20 + \varepsilon} |T_1|}\right)$ or $\mathcal{O}\left(\frac{|T_1||T_2|}{\log ^ {20 + \varepsilon} |T_2|}\right)$) for any $\varepsilon > 0$.
\end{cor}

The same hardness results for PMLG apply for several more specific types of graphs (details will be presented in the full version of this paper). These include when the graph $G$ is a deterministic DAG (at most one edge leaves a vertex with the same leading character on an edge label) of total degree at most 3, and the case when $G$ is a directed or undirected planar graph of degree at most $3$.

\section{Reduction from Formula-SAT to PMLG}
\label{sec:PMLG}
\subsection{Technical Overview}
Our reduction from Formula-SAT to PMLG uses an intermediate problem called Formula-Pair.
\begin{defn}[Formula-Pair] Given a deMorgan Formula $F = F(x_1, \dots, x_m, y_1, \dots, y_m)$ of size $2m$ where each input is used exactly once, and two sets $A, B \subseteq \{0,1\}^m$ each of size $N$, does there exist $a \in A$ and $b \in B$ such that $F(a,b) = F(a_1, \dots, a_m, b_1, \dots, b_m) = 1$?
\end{defn}

The role Formula-Pair plays in our reduction is analogous to the role of the Orthogonal Vectors Problem in many SETH reductions. It was proven in \cite{DBLP:conf/icalp/AbboudB18} that an instance of Formula-SAT on a formula of size $s$ over $n$ inputs can be reduced to an instance of Formula-Pair on two sets of size $N = \mathcal{O}(2^{n/2})$ and a formula of size $\mathcal{O}(s)$ in linear time  (in particular, they reduce from a harder problem they call $\mathcal{F}_1$-Formula-SAT). 
Note that we may assume that $F$ contains no input gates with negated binary variables, 
since if variable $x_i$ is negated in $F$, we can flip bit $a_i$ for all $a \in A$.

We begin our reduction from Formula-Pair to PMLG by considering a formula $F$ and some input bit assignments $a \in A$ and $b \in B$. We then construct a pattern $P$ and labeled graph $G$ such that $P$ occurs in $G$ if and only if together $a$ and $b$ satisfy $F$. In this step, we must ensure that our construction of $P$ only relies on the input bit assignments of $a$, and our construction of $G$ only relies on the input bit assignments of $b$. This allows us to create patterns $P_1, P_2, \dots, P_N$ corresponding to the $N$ bit assignments in $A$, and graphs $G_1, G_2, \dots, G_N$ corresponding to the $N$ bit assignments in $B$. Then we will have that $P_i$ occurs in $G_j$ if and only if $F(a, b) = 1$, where $a \in A$ is the bit assignment corresponding to $P_i$, and $b \in B$ is the  bit assignment corresponding to $G_j$. Finally, we combine these patterns and graphs into a product pattern $P$ and a product graph $G$ such that $P$ occurs in $G$ if and only if some $P_i$ occurs in some $G_j$. This will complete the reduction.

\subsection{Reduction}
Given a deMorgan formula $F$ and a complete assignment of input bits $(a, b)$ where $a \in A$ and $b \in B$, we will construct a corresponding pattern $P$ and labeled DAG $G$ over alphabet $\{0, 1, \$\}$ such that $P$ occurs in $G$ if and only if the output of $F$ is $1$ on input $(a, b)$. This pattern and graph will be built recursively, starting with the input gates as a base case. For a gate $g = (g_1 * g_2)$ where $* \in \{\lor, \land\}$, we will construct a corresponding pattern and graph for gate $g$ by merging the patterns and graphs of subgates $g_1$ and $g_2$. At each step in this process, the pattern corresponding to gate $g$ occurs in the graph corresponding to gate $g$ if and only if $g$ evaluates to $1$ on input $(a, b)$.

\bigskip
\noindent
\textbf{Invariants.} 
We will maintain the following invariants during this recursive procedure. Let $g$ be a gate of $F$ with height $h$, and let $P$ and $G$ be the pattern and graph corresponding to gate $g$ in our construction. 

\begin{enumerate}
    \item Graph $G$ will have a designated source vertex and sink vertex, both with label ``$1$''.  Every maximal path in $G$ will be of length $|P|$ and start and end at the source and sink vertices of $G$ respectively.
    \item The construction of pattern $P$ is independent of the choice of bit assignment $b \in B$, and the construction of graph  $G$ is independent of the choice of bit assignment $a \in A$.  
    \item Pattern $P$ occurs in $G$ if and only if $g$ has output $1$ on input $(a, b)$.
\end{enumerate}

Observe that by the first invariant, every occurrence of pattern $P$ in graph $G$ will start at the source vertex of $G$ and end at the sink vertex of $G$. If this is the case, we will say that $G$ matches $P$. We will also refer to the designated source and sink vertices of $G$ as the start and end vertices of $G$.

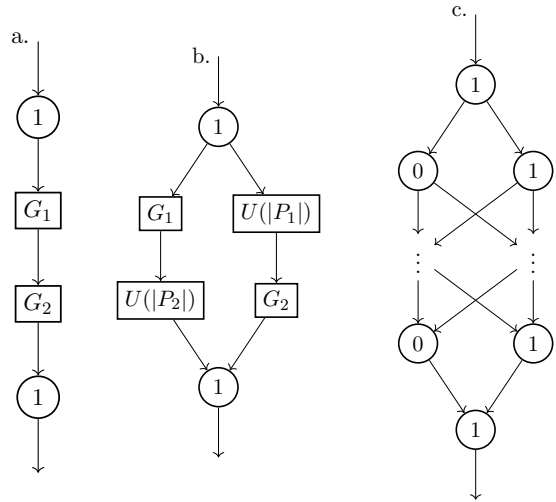
\begin{figure}[H]
\centering
\begin{minipage}[l]{.06\textwidth}
\centering
\resizebox{\textwidth}{!}{
\begin{tikzpicture}[
  emptynode/.style = {draw=none, minimum size=5.5mm},
  lightnode/.style = {circle, draw=black!255, thick, minimum size=5.5mm},
  squarenode/.style={draw=black!255, thick, minimum size=5mm},
  level 1/.style = {sibling distance=6cm},
  level 2/.style = {sibling distance=3cm},
  level 3/.style = {sibling distance=2cm},
  level 4/.style = {sibling distance=.8cm},
]
\node[minimum width=0cm,anchor = north east] {a.};
\node[emptynode] {} 
child[->,line width=.5pt]
    {
    node[lightnode]{$1$}
    child[->,line width=.5pt]
        {
        node[squarenode]{$G_1$}
        child[->,line width=.5pt]
            {
            node[squarenode]{$G_2$}
            child[->,line width=.5pt]
                {
                node[lightnode]{$1$}
                child[->,line width=.5pt]
                    {
                    node[emptynode]{}
                    }
                }
            }
        }
    };
\end{tikzpicture}
}
\end{minipage}
\begin{minipage}[r]{.2\textwidth}
\centering
\resizebox{\textwidth}{!}{
\begin{tikzpicture}[
  emptynode/.style = {draw=none, minimum size=5.5mm},
  lightnode/.style = {circle, draw=black!255, thick, minimum size=5.5mm},
  squarenode/.style={draw=black!255, thick, minimum size=5mm},
  level 1/.style = {sibling distance=6cm},
  level 2/.style = {sibling distance=2cm},
  level 3/.style = {sibling distance=2cm},
  level 4/.style = {sibling distance=2cm},
]
\node[minimum width=0cm,anchor = north east] {b.};
\node[emptynode] {} 
child[->,line width=.5pt]
    {
    node[lightnode]{$1$}
    child[->,line width=.5pt]
        {
        node[squarenode]{$G_1$}
        child[->,line width=.5pt]
            {
            node[squarenode]{$U(|P_2|)$}
            child[]
                {
                node[emptynode]{}
                edge from parent[draw=none]
                }
            child[->,line width=.5pt]
                {
                node[lightnode](1){$1$}
                child[->,line width=.5pt]
                    {
                    node[emptynode]{}
                    }
                }
            }
        }
    child[->,line width=.5pt]
        {
        node[squarenode]{$U(|P_1|)$}
        child[->,line width=.5pt]
            {
            node[squarenode](G2){$G_2$}
            }
        }
    };
\draw[->] (G2) -- (1);
\end{tikzpicture}
}
\end{minipage}
\begin{minipage}[r]{.18\textwidth}
\centering
\resizebox{\textwidth}{!}{
\begin{tikzpicture}[
  emptynode/.style = {draw=none, minimum size=5.5mm},
  lightnode/.style = {circle, draw=black!255, thick, minimum size=5.5mm},
  squarenode/.style={draw=black!255, thick, minimum size=5mm},
  level 1/.style = {sibling distance=6cm},
  level 2/.style = {sibling distance=2cm},
  level 3/.style = {sibling distance=2cm},
  level 4/.style = {sibling distance=2cm},
  level 5/.style = {sibling distance=2cm},
  level 6/.style = {sibling distance=2cm},
]
\node[minimum width=0cm,anchor = north east] {c.};
\node[emptynode](0) {} 
child[->,line width=.5pt]
    {
    node[lightnode](1){$1$}
    child[->,line width=.5pt]
        {
        node[lightnode](2){$0$}
        child[->,line width=.5pt]
                {
                node[->,emptynode](4){$\vdots$}
                    child[->,line width=.5pt]
                    {
                    node[lightnode](5){$0$}
                    child[->,line width=.5pt]
                        {
                        node[emptynode](6){}
                        edge from parent[draw=none]
                        }
                    child[->,line width=.5pt]
                        {
                        node[lightnode](7){$1$}
                        child[->,line width=.5pt]
                            {
                            node[emptynode](8){}
                            }
                        }
                    }
                }
        }
    child[->,line width=.5pt]
        {
        node[lightnode](9){$1$}
        child[->,line width=.5pt]
                {
                node[emptynode](11){$\vdots$}
                child[->,line width=.5pt]
                    {
                    node[lightnode](12){$1$}
                    }
                }
        }
    };
\draw[->] (2) -- (11);
\draw[->] (9) -- (4);
\draw[->] (4) -- (12);
\draw[->] (11) -- (5);
\draw[->] (12) -- (7);
\end{tikzpicture}
}
\end{minipage}

\caption{
From left to right: the graph constructed for gate $g = (g_1 \land g_2)$, the graph constructed for gate $g = (g_1 \lor g_2)$, and the Universal Subgraph $U(x)$. Note that Universal Subgraph $U(x)$ has a series of $x - 2$ vertex pairs labeled $0$ and $1$, so that its  maximal path length is $x$.
}
\label{fig:pmlg_gates}
\end{figure}

\bigskip
\noindent
\textbf{Input Gate.}
Each input gate $g$ in $F$ takes as input a binary variable $z$. We will design a graph $G$ and pattern $P$ such that $G$ matches $P$ if and only if $z$ had value $1$ in bit assignment $(a, b)$, and hence $g$ evaluates to $1$. Our construction depends on whether $z$ corresponds to an input bit in $a$ or $b$.
\begin{itemize}
\item \textbf{Case 1. $\bm{z}$ corresponds to some $\bm{a_i \in a}$.} 
We let $P := 1a_i1$ and $G$ be a path of length three with all vertices labeled $1$.   
\item \textbf{Case 2. $\bm{z}$ corresponds to some $\bm{b_i \in b}$.}
We let $P := 111$ and $G$ be a path of length three with the first and last vertex labeled $1$ and the middle vertex labeled $b_i$.
\end{itemize}

The start vertex of $G$ will be the first vertex in the path, and the end vertex of $G$ will be the third (last) vertex in the path. Then our graph $G$ matches pattern $P$ if and only if $z = 1$ and thus the input gate evaluates to true.  Additionally, the construction of $P$  does not depend on $b$ and the construction of $G$ does not depend on $a$. All invariants are satisfied.

\bigskip
\noindent
\textbf{AND Gate.} Given a gate $g = (g_1 \land g_2)$ and the graphs and patterns $(G_1, P_1)$ and $(G_2, P_2)$ corresponding to gates $g_1$ and $g_2$ respectively, we must construct a product graph $G$ and pattern $P$ such that $G$ matches $P$ if and only if $G_1$ matches $P_1$ and $G_2$ matches $P_2$. This is done rather easily. Let $P := 1P_1 P_21$. Now let our product graph $G$ be defined as in Figure \ref{fig:pmlg_gates}.a.
Our start vertex is labeled $1$ and has an outgoing edge to the start vertex of subgraph $G_1$. The end vertex of $G_1$ in turn has an outgoing edge to start vertex of subgraph $G_2$, whose own end vertex has an outgoing edge to the final vertex of $G$. We now verify all invariants are satisfied. 

\begin{itemize}
    \item \textbf{Invariant 1.} We assume that every maximal path in $G_1$ (respectively $G_2$) is of length $|P_1|$ (respectively $|P_2|$). Then by the construction of $P$ and $G$, every maximal path in $G$ is of length $|P|$. The invariant is maintained.
    \item \textbf{Invariant 2.} Assuming that the construction of $P_1$ and $P_2$ is independent of $b$, and the construct of $G_1$ and $G_2$ is independent of $a$, it follows that the construction of pattern $P$ is independent of bit assignment $b$, and the construction of graph $G$ is independent of bit assignment $a$.
    \item \textbf{Invariant 3.} Since every occurrence of $P$ in $G$ starts at the start vertex of $G$ and ends at the end vertex, we must conclude that $P$ occurs in $G$ if and only if $P_1$ occurs in $G_1$ and $P_2$ occurs in $G_2$. Then by our invariant $P$ occurs in $G$ if and only if $g$ evaluates to $1$ on input $(a, b)$. The invariant is preserved. 
\end{itemize}
\noindent
\textbf{OR Gate.}
Given a gate $g = (g_1 \lor g_2)$ and the graphs and patterns $(G_1, P_1)$ and $(G_2, P_2)$ corresponding to gates $g_1$ and $g_2$ respectively, we must construct a product graph $G$ and pattern $P$ such that $G$ matches $P$ if and only if $G_1$ matches $P_1$ or $G_2$ matches $P_2$. As with our AND gate, we let $P := 1P_1P_21$. Our product graph $G$ (see Figure \ref{fig:pmlg_gates}.b) splits into two branches. One branch checks if $G_1$ matches $P_1$ and ignores $P_2$, while the other branch checks if $G_2$ matches $P_2$ and ignores $P_1$. We are able to ignore $P_2$ (respectively $P_1$) by constructing a `universal' subgraph that matches all binary strings that start and end with $1$ and are of length $|P_2|$ (respectively $|P_1|$). We let $U(x)$ denote the universal subgraph for length $x$, and we depict our construction of $U(x)$ in Figure \ref{fig:pmlg_gates}.c. Observe that graphs $U(|P_1|)$ and $U(|P_2|)$ match $P_1$ and $P_2$ respectively. We now check that all invariants are satisfied.



\begin{itemize}
    \item \textbf{Invariant 1.}  A similar argument as in the AND gate shows that every maximal path in $G$ is of length $|P|$ and passes through the start and end vertices of $G$. The invariant is preserved.
    \item \textbf{Invariant 2.} Pattern $P$ is independent of bit assignment $b$ by a similar argument as with the AND gate construction. However, for our graph $G$, we must verify that subgraphs $U(|P_1|)$ and $U(|P_2|)$ of $G$ do not depend on bit assignment $a$. This will follow from proving that the lengths of patterns $P_1$ and $P_2$ do not depend on the bit assignment $a$. Note that in each of the input, AND, and OR gate constructions, the length of the constructed pattern is the same regardless of the bit assignment $a$. Thus we conclude that $U(|P_1|)$ and $U(|P_2|)$ are independent of the bit assignment $a$, and therefore the construction of graph $G$ is independent of the bit assignment $a$.
   \item \textbf{Invariant 3.} Since every occurrence of pattern $P$ starts at the start vertex of $G$ and ends at the end vertex, it is immediate that $G$ matches $P$ if and only if $G_1$ matches $P_1$ or $G_2$ matches $P_2$. It immediately follows from our invariant that $G$ matches $P$ if and only if gate $g = (g_1 \lor g_2)$ evaluates to $1$ on input $(a, b)$.
\end{itemize}

\subsection{Completing the Reduction}
Now corresponding to our formula $F$ of size $s$ and a complete assignment of input bits $(a, b)$, we can build a pattern $P$ and a graph $G$ such that $G$ matches $P$ if and only if assignment $(a, b)$ satisfies $F$. Note that we only add a constant number of symbols to our pattern $P$ for each gate in $F$, and there are fewer than $2s$ gates in $F$, so $|P| = \mathcal{O}(s)$. On the other hand, each OR gate in $F$ can contribute $\mathcal{O}(|P|)$ vertices and edges to our final graph $G$. It follows that $G$ is of size $\mathcal{O}(s^2)$. 

Using our construction, for every $a \in A$ we may construct a corresponding pattern $P$, and for every $b \in B$ we may construct a corresponding graph $G$. We will denote these patterns and graphs by $P_1, P_2, \dots, P_N$ and $G_1, G_2, \dots, G_N$  respectively. 
Note that each pattern $P_j$ makes no assumptions on the bit assignment $b$, and graph $G_i$ makes no assumptions on the bit assignment $a$. It follows that $G_i$ matches $P_j$ if and only if together the corresponding bit assignments $a \in A$ and $b \in B$ satisfy $F$. 

\begin{figure}[ht]
\centering
\resizebox{\textwidth}{!}{
\begin{tikzpicture}[
    thick, shorten >=1pt, 
    auto,
    node distance=1.6cm,
  main node/.style={circle, draw, fill=none},
  empty node/.style={draw=none},
  squarenode/.style={draw=black!255, thick, minimum size=5mm},
]
\node[main node] (r2_1)[] {$\$$};
\node[squarenode] (r2_2) [right of = r2_1, label=above:{$1$}] {$U(\mu)$};
\node[main node] (r2_3) [right of = r2_2] {$\$$};

\node[main node] (r2_4) [right of = r2_3] {$\$$};
\node[squarenode] (r2_5) [right of = r2_4, label=above:{$N-1$}] {$U(\mu)$};
\node[main node] (r2_6) [right of = r2_5] {$\$$};

\node[main node] (r2_7) [right of = r2_6] {$\$$};
\node[squarenode] (r2_8) [right of = r2_7, label=above:{$2N-2$}] {$U(\mu)$};
\node[main node] (r2_9) [right of = r2_8] {$\$$};

\node[main node] (r1_1)[above of = r2_1] {$\$$};
\node[main node] (r1_2)[above of = r2_4] {$\$$};
\node[main node] (r1_3)[above of = r2_7] {$\$$};

\node[main node] (r3_1) [below right of = r2_6] {$\$$};
\node[main node] (r3_2) [below right of = r2_9] {$\$$};

\node[main node] (r4_1) [below of = r3_1] {$\$$};
\node[main node] (r4_2) [below of = r3_2] {$\$$};

\node[squarenode] (r5_1) [below of = r4_1] {$G_1$};
\node[main node] (r5_2) [right of = r5_1] {$\$$};
\node[squarenode] (r5_3) [below of = r4_2] {$G_N$};
\node[main node] (r5_4) [right of = r5_3] {$\$$};

\node[main node] (r6_1) [below right of = r5_2] {$\$$};
\node[main node] (r6_2) [below right of = r5_4] {$\$$};

\node[squarenode] (r7_1) [below left of = r6_1, label=below:{$1$}] {$U(\mu)$};
\node[main node] (r7_2) [right of = r7_1] {$\$$};
\node[main node] (r7_3) [right of = r7_2] {$\$$};
\node[squarenode] (r7_4) [below left of = r6_2, label=below:{$N$}] {$U(\mu)$};
\node[main node] (r7_5) [right of = r7_4] {$\$$};
\node[main node] (r7_6) [right of = r7_5] {$\$$};
\node[squarenode] (r7_7) [right of = r7_6, label=below:{$2N-2$}] {$U(\mu)$};
\node[main node] (r7_8) [right of = r7_7] {$\$$};

\node[main node] (r8_1) [below right of = r7_2] {$\$$};
\node[main node] (r8_2) [below right of = r7_5] {$\$$};
\node[main node] (r8_3) [below right of = r7_8] {$\$$};

\path[]
    (r2_1) edge [->](r2_2) 
    (r2_2) edge [->] (r2_3)
    (r2_3) edge [loosely dotted] (r2_4)
    (r2_4) edge [->](r2_5) 
    (r2_5) edge [->] (r2_6)
    (r2_6) edge [loosely dotted] (r2_7)
    (r2_7) edge [->](r2_8) 
    (r2_8) edge [->] (r2_9)
    
    (r1_1) edge [->](r2_1) 
    (r1_2) edge [->] (r2_4)
    (r1_3) edge [->] (r2_7)
    
    (r3_1) edge [->] (r4_1)
    (r3_2) edge [->] (r4_2)
    
    (r4_1) edge [->] (r5_1)
    (r4_2) edge [->] (r5_3)
    
    (r2_6) edge [bend right, ->] (r5_1)
    (r2_9) edge [bend right, ->] (r5_3)
    
    (r5_1) edge [->] (r5_2)
    (r5_2) edge [loosely dotted] (r5_3)
    (r5_3) edge [->] (r5_4)
    
    (r5_2) edge [->] (r6_1)
    (r5_4) edge [->] (r6_2)
    
    (r5_2) edge [->] (r7_1)
    (r5_4) edge [->] (r7_4)
    
    (r7_1) edge [->] (r7_2)
    (r7_2) edge [loosely dotted] (r7_3)
    (r7_3) edge [->] (r7_4)
    (r7_4) edge [->] (r7_5)
    (r7_5) edge [loosely dotted] (r7_6)
    (r7_6) edge [->] (r7_7)
    (r7_7) edge [->] (r7_8)
    
    (r7_2) edge [->] (r8_1)
    (r7_5) edge [->] (r8_2)
    (r7_8) edge [->] (r8_3);
\end{tikzpicture}
}
\caption{Our final graph $G$. Here $\mu = |P_i|$.}
\label{fig:pmlg_final}
\end{figure}
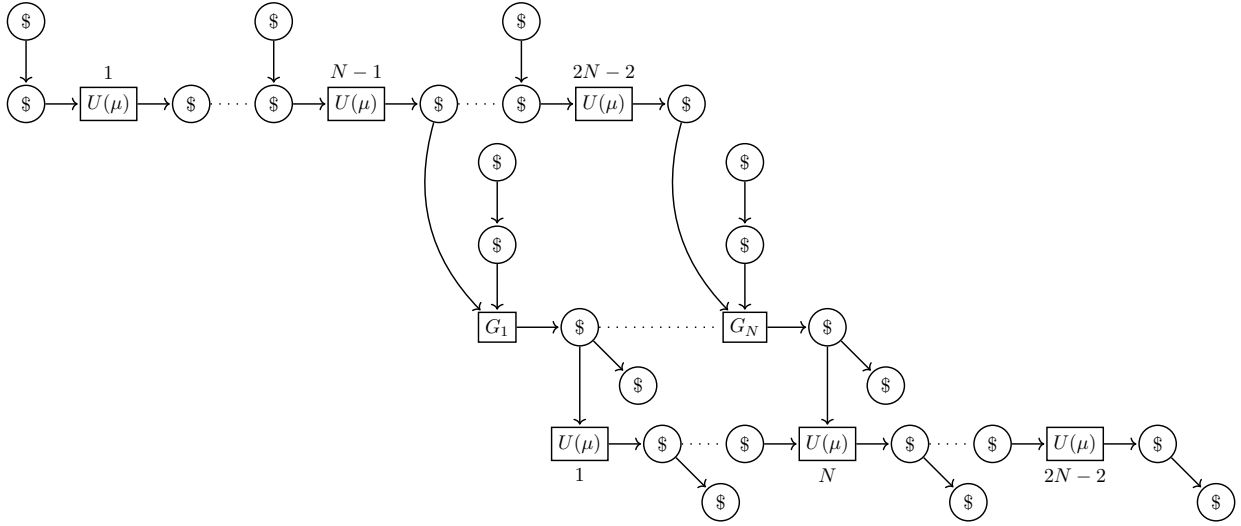

Next, we construct a final graph $G$ and pattern $P$ such that $P$ occurs in $G$ if and only if some $G_i$ matches some $P_j$. This will complete our reduction. We define our final pattern $P$ as follows:
$P := \$\$P_1\$P_2 \$\cdots \$ P_N \$\$$.
The structure of our final graph $G$ 
is similar to the final graph presented in \cite{DBLP:conf/icalp/EquiGMT19}. We present this graph in Figure \ref{fig:pmlg_final} and briefly explain the intuition behind it.
Let $\mu = |P_i|$ for any $i$. Then subgraph $U(\mu)$ will match any subpattern $P_i$ in $P$. The graph $G$ uses $U(\mu)$ to match the subpatterns $P_i$ in $P$ that do not match with any $G_j$. Note that since pattern $P$ has a prefix of two $\$$ symbols and a suffix of two $\$$ symbols, $P$ is forced to pass through the second row of $G$. More specifically, the first row of $G$ alone cannot match the $\$\$$ suffix of $P$, and the third row of $G$ alone cannot match the $\$\$$ prefix of $P$. Then it can be seen that $P$ occurs in $G$ only if $P$ passes through the second row of $G$, and hence some subgraph $G_i$ matches some subpattern $P_j$. Then by construction, $P$ occurs in $G$ if and only if there exists  $a \in A$ and $b \in B$ such that $F(a, b) = 1$. Furthermore, our final graph is a DAG of size $\mathcal{O}(N \cdot s^2)$ and our final pattern $P$ is of length $\mathcal{O}(N \cdot s)$. This completes our reduction from Formula-SAT to PMLG on DAGs.

\section{Reduction from Formula-SAT to Subtree Isomorphism}
\label{sec:subtree}
\subsection{Technical Overview}
We begin our reduction from Formula-Pair to Subtree Isomorphism by considering a formula $F$ and some input bit assignments $a \in A$ and $b \in B$. We then construct trees $T_a$ and $T_b$ such that $T_a$ is contained in $T_b$ if and only if together $a$ and $b$ satisfy $F$. In this step it is important that we ensure that our construction of $T_a$ only relies on the input bit assignments of $a$, and our construction of $T_b$ only relies on the input bit assignments of $b$. This allows us to create $N$ $T_a$ trees corresponding to the $N$ bit assignments $a$ in $A$, and $N$ $T_b$ trees corresponding to the $N$ bit assignments $b$ in $B$. Then we will have that some $T_a$ tree is contained in  some $T_b$ tree if and only if the corresponding bit assignments $a \in A$ and $b \in B$ satisfy  $F(a, b) = 1$. Finally, we combine these trees into two final trees $T_A$ and $T_B$ such that $T_A$ is contained in $T_B$ if and only if some $T_{a}$ is contained in some $T_{b}$. This will complete the reduction.

\subsection{Reduction}
Given a deMorgan formula $F$ and a complete assignment of input bits $(a, b)$ where $a \in A$ and $b \in B$, we will construct the corresponding rooted trees $T_{a}$ and $T_{b}$ such that $T_{a}$ is contained in $T_{b}$ if and only if the output of $F(a, b) = 1$. These trees will be constructed recursively, starting with the input gates of $F$ as a base case. For a gate $g = (g_1 * g_2)$ where $* \in \{\lor, \land\}$, we will construct the corresponding trees $T_a^g$ and $T_b^g$ for gate $g$ by merging the trees of subgates $g_1$ and $g_2$. At each step in this process, $T_a^g$ will be contained in $T_b^g$  if and only if gate $g$ has output $1$ on input $(a, b)$.

\bigskip
\noindent
\textbf{Invariants.} We will maintain the following invariants throughout our construction. Let $g$ be a gate of $F$ with height $h$.
\begin{enumerate}
    \item The height of $T_a^g$ is equal to the height of $T_b^g$ and is at most $4h$.
    
    \item The construction of $T_a^g$ is independent of the choice of bit assignment $b \in B$, and the construction of $T_b^g$ is independent of the choice of bit assignment $a \in A$.
    
    \item Tree $T_a^g$ is contained in tree $T_b^g$ if and only if gate $g$ has output $1$ on input $(a, b)$.
\end{enumerate}

\begin{center}
\begin{minipage}{.47\textwidth}
\label{fig:input_gate}
\centering
\resizebox{.77\textwidth}{!}{
\centering
\begin{tabular}{ccc}
   Input  &  ~~~~~~~$T_a^g$ & ~~~~~$T_b^g$ \\
   \hline
   &&\\
   &&\\
   &&\\
   \vspace{-4.5em} $a_i = 0$ & & \\
   & 
   \begin{tikzpicture}[
    lightnode/.style = {circle, fill = black, inner sep=2pt},
    level 1/.style = {sibling distance=20mm},
    ]

        \node[lightnode, label=left:{\Large $v_a$}] {}
            child[line width=.5pt] {
                node[lightnode, label=left:{}] {}
                edge from parent[solid]
            }
            child[line width=.5pt] {
                node[lightnode, label=left:{}] {}
                edge from parent[solid]
            };
   \end{tikzpicture}
   &
    \begin{tikzpicture}[
    lightnode/.style = {circle, fill = black, inner sep=2pt, solid},
    level 1/.style = {sibling distance=20mm},
    ]

        \node[lightnode, label=left:{\Large $v_b$}] {}
            child[line width=.5pt] {
                node[lightnode, label=left:{}] {}
                edge from parent[solid]
            };
   \end{tikzpicture} \\

   \hline
   &&\\
   &&\\
   &&\\
   \vspace{-4.5em}
      $a_i = 1$ & \\
      &
   \begin{tikzpicture}[
    lightnode/.style = {circle, fill = black,inner sep=2pt, solid},
    level 1/.style = {sibling distance=20mm},
    ]
    \node[lightnode, label=left:{\Large $v_a$}] {}
        child[line width=.5pt] {
            node[lightnode, label=left:{}] {}
            edge from parent[solid]
        };
   \end{tikzpicture}
   &
    \begin{tikzpicture}[
    lightnode/.style = {circle, fill = black, inner sep=2pt, solid},
    level 1/.style = {sibling distance=20mm},
    ]

    \node[lightnode, label=left:{\Large $v_b$}] {}
       child[line width=.5pt] {
            node[lightnode, label=left:{}] {}
            edge from parent[solid]
        };
   \end{tikzpicture} \\
   
    \hline
   &&\\
   &&\\
   &&\\
   \vspace{-4.5em} $b_j = 0$ & & \\
   & 

   \begin{tikzpicture}[
    lightnode/.style = {circle, fill = black, inner sep=2pt, solid},
    level 1/.style = {sibling distance=20mm},
    ]

    \node[lightnode, label=left:{\Large $v_a$}] {}
        child[line width=.5pt] {
            node[lightnode, label=left:{}] {}
            edge from parent[solid]
        }
        child[line width=.5pt] {
            node[lightnode, label=left:{}] {}
            edge from parent[solid]
        };
   \end{tikzpicture}
    &
    \begin{tikzpicture}[
    lightnode/.style = {circle, fill = black, inner sep=2pt, solid},
    level 1/.style = {sibling distance=20mm},
    ]

    \node[lightnode, label=left:{\Large $v_b$}] {}
        child[line width=.5pt] {
            node[lightnode, label=left:{}] {}
            edge from parent[solid]
        };
   \end{tikzpicture} \\
   \hline
   &&\\
   &&\\
   &&\\
      \vspace{-4.5em} $b_j = 1$ & & \\
      &
   \begin{tikzpicture}[
    lightnode/.style = {circle, fill = black, inner sep=2pt, solid},
    level 1/.style = {sibling distance=20mm},
    ]
        \node[lightnode, label=left:{\Large $v_a$}] {}
            child[line width=.5pt] {
                node[lightnode, label=left:{}] {}
                edge from parent[solid]
            }
            child[line width=.5pt] {
                node[lightnode, label=left:{}] {}
                edge from parent[solid]
            };
   \end{tikzpicture}
   &
    \begin{tikzpicture}[
    lightnode/.style = {circle, fill = black, inner sep=2pt, solid},
    level 1/.style = {sibling distance=20mm},
    ]
        \node[lightnode, label=left:{\Large $v_b$}] {}
            child[line width=.5pt] {
                node[lightnode, label=left:{}] {}
                edge from parent[solid]
            }
            child[line width=.5pt] {
                node[lightnode, label=left:{}] {}
                edge from parent[solid]
            };
   \end{tikzpicture} \\
   \hline
\end{tabular}
}

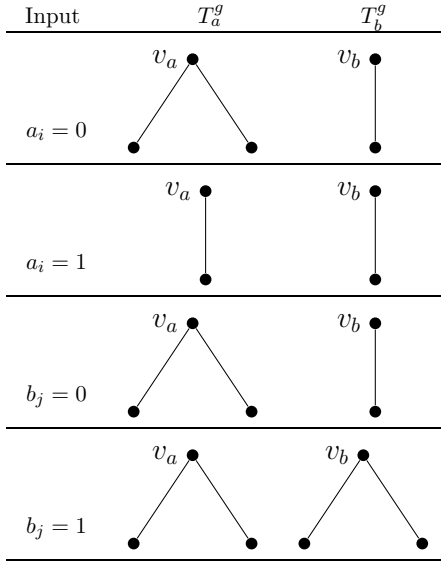
\captionof{figure}{The trees $T_a^g$ and $T_b^g$ corresponding to input gate $g = a_i$ or $g = b_j$.}
\label{fig:subtree_input}
\end{minipage}
\begin{minipage}{.47\textwidth}
\begin{minipage}{1\textwidth}
\resizebox{\textwidth}{!}{
\begin{tikzpicture}[
  lightnode/.style = {circle, fill = black, inner sep=2pt, solid},
  emptynode/.style={},
  level 1/.style = {level distance=12mm, sibling distance=1.2cm},
  level 2/.style = {level distance=12mm, sibling distance=1.2cm},
  level 3/.style = {level distance=12mm, sibling distance=1.2cm},
  level 4/.style = {level distance=12mm, sibling distance=1.2cm},
  triangle/.style={
  draw,solid,shape border uses incircle,
  isosceles triangle,shape border rotate=90,yshift=-6.5mm, minimum width=2cm},
]

\node[lightnode, label=above:{\Large $v_a^0$}] {}
    child[line width=.5pt]{
        node[emptynode] (r2_1) {}
        edge from parent[draw = none]
    }
    child[line width=.5pt]{
        node[emptynode] (r2_2) {}
        edge from parent[draw = none]
    }
    child[line width=.5pt]{
        node[lightnode, label=above:{\Large $1$}] (r2_3) {}
        edge from parent[draw = none]
    }
    child[line width=.5pt] {
        node[lightnode, label=right:{\Large $v_a^1$}] (r2_4) {}
    }
    child[line width=.5pt] {
        node[emptynode] (r2_5) {}
        child[line width=.5pt]{
            node[emptynode] (r3_1) {}
            edge from parent[draw = none]
        }
        child[line width=.5pt]{
            node[lightnode, label=above:{\Large $2$}] (r3_2) {}
            edge from parent[draw = none]
        }
        child[line width=.5pt]{
            node[lightnode] (r3_3) {}
            edge from parent[draw = none]
        }
        child[line width=.5pt] {
            node[lightnode, label=right:{\Large $v_a^3$}] (r3_4) {}
            child[line width=.5pt] {
                node[lightnode] (r4_1) {}
                child[line width=.5pt]{
                    node[triangle] {\Large $T_a^1$}
                }
            }
            edge from parent[draw = none]
        }
        child[line width=.5pt] {
            node[emptynode] (r3_5) {}
            edge from parent[draw = none]
        }
        child[line width=.5pt] {
        node[lightnode, label=left:{\Large $v_a^4$}] (r3_6) {}
        child[line width=.5pt] {
                node[lightnode] (r4_2) {}
                child[line width=.5pt]{
                    node[triangle] {\Large $T_a^2$}
                }
            }
        edge from parent[draw = none]
        }
        child[line width=.5pt]{
            node[lightnode,  label=above:{\Large $1$}] (r3_7) {}
            edge from parent[draw = none]
        }
        child[line width=.5pt]{
            node[emptynode] (r3_8) {}
            edge from parent[draw = none]
        }
        child[line width=.5pt] {
            node[emptynode] (r3_9) {}
            edge from parent[draw = none]
        }
        edge from parent[draw = none]
    }
    child[line width=.5pt] {
        node[lightnode, label=left:{\Large $v_a^2$}] (r2_6) {}
    }
    child[line width=.5pt]{
        node[lightnode] (r2_7) {}
        edge from parent[draw = none]
    }
    child[line width=.5pt]{
        node[lightnode,label=above:{\Large $2$}] (r2_8) {}
        edge from parent[draw = none]
    }
    child[line width=.5pt] {
        node[emptynode] (r2_9) {}
        edge from parent[draw = none]
    };

\draw[solid] (r2_3) to (r2_4);
\draw[solid] (r2_6) to (r2_7);
\draw[solid] (r2_7) to (r2_8);

\draw[solid] (r2_4) to (r3_4);
\draw[solid] (r2_6) to (r3_6);

\draw[solid] (r3_2) to (r3_3);
\draw[solid] (r3_3) to (r3_4);
\draw[solid] (r3_6) to (r3_7);
\end{tikzpicture}
}
\vspace{1mm}
\end{minipage}
\begin{minipage}{1\textwidth}
\resizebox{\textwidth}{!}{
\begin{tikzpicture}[
  lightnode/.style = {circle, fill = black,inner sep=2pt, solid},
  emptynode/.style={},
  level 1/.style = {level distance=12mm, sibling distance=1.2cm},
  level 2/.style = {level distance=12mm, sibling distance=1.2cm},
  level 3/.style = {level distance=12mm, sibling distance=1.2cm},
  level 4/.style = {level distance=12mm, sibling distance=1.2cm},
  triangle/.style={
  draw,solid,shape border uses incircle,
  isosceles triangle,shape border rotate=90,yshift=-6.5mm, minimum width=2cm},
]

\node[lightnode, label=above:{\Large $v_b^0$}] {}
    child[line width=.5pt]{
        node[emptynode] (r2_1) {}
        edge from parent[draw = none]
    }
    child[line width=.5pt]{
        node[emptynode] (r2_2) {}
        edge from parent[draw = none]
    }
    child[line width=.5pt]{
        node[lightnode, label=above:{\Large $1$}] (r2_3) {}
        edge from parent[draw = none]
    }
    child[line width=.5pt] {
        node[lightnode, label=right:{\Large $v_b^1$}] (r2_4) {}
    }
    child[line width=.5pt] {
        node[emptynode] (r2_5) {}
        child[line width=.5pt]{
            node[emptynode] (r3_1) {}
            edge from parent[draw = none]
        }
        child[line width=.5pt]{
            node[lightnode,label=above:{\Large $2$}] (r3_2) {}
            edge from parent[draw = none]
        }
        child[line width=.5pt]{
            node[lightnode] (r3_3) {}
            edge from parent[draw = none]
        }
        child[line width=.5pt] {
            node[lightnode, label=right:{\Large $v_b^3$}] (r3_4) {}
            child[line width=.5pt] {
                node[lightnode] (r4_1) {}
                child[line width=.5pt]{
                    node[triangle] {\Large $T_b^1$}
                }
            }
            edge from parent[draw = none]
        }
        child[line width=.5pt] {
            node[emptynode] (r3_5) {}
            edge from parent[draw = none]
        }
        child[line width=.5pt] {
        node[lightnode, label=left:{\Large $v_b^4$}] (r3_6) {}
        child[line width=.5pt] {
                node[lightnode] (r4_2) {}
                child[line width=.5pt]{
                    node[triangle] {\Large $T_b^2$}
                }
            }
        edge from parent[draw = none]
        }
        child[line width=.5pt]{
            node[lightnode,label=above:{\Large $1$}] (r3_7) {}
            edge from parent[draw = none]
        }
        child[line width=.5pt]{
            node[emptynode] (r3_8) {}
            edge from parent[draw = none]
        }
        child[line width=.5pt] {
            node[emptynode] (r3_9) {}
            edge from parent[draw = none]
        }
        edge from parent[draw = none]
    }
    child[line width=.5pt] {
        node[lightnode, label=left:{\Large $v_b^2$}] (r2_6) {}
    }
    child[line width=.5pt]{
        node[lightnode] (r2_7) {}
        edge from parent[draw = none]
    }
    child[line width=.5pt]{
        node[lightnode,label=above:{\Large $2$}] (r2_8) {}
        edge from parent[draw = none]
    }
    child[line width=.5pt] {
        node[emptynode] (r2_9) {}
        edge from parent[draw = none]
    };

\draw[solid] (r2_3) to (r2_4);
\draw[solid] (r2_6) to (r2_7);
\draw[solid] (r2_7) to (r2_8);

\draw[solid] (r2_4) to (r3_4);
\draw[solid] (r2_6) to (r3_6);

\draw[solid] (r3_2) to (r3_3);
\draw[solid] (r3_3) to (r3_4);
\draw[solid] (r3_6) to (r3_7);
\end{tikzpicture}
}
\end{minipage}
\captionof{figure}{The trees $T_a^g$ (top) and $T_b^g$ (bottom) corresponding to AND gate $g = (g_1 \land g_2)$.}
\label{fig:subtree_AND}
\end{minipage}
\end{center}

\bigskip
\noindent
\textbf{Input Gate.} 
Given an input gate $g$ corresponding to a bit value $a_i \in a$ (respectively, a bit value $b_j \in b$), we will construct trees $T_a^g$ and $T_b^g$ so that $T_a^g$ is contained in $T_b^g$ if and only if $a_i = 1$ (respectively, $b_j = 1$). We construct $T_a^g$ and $T_b^g$ as in Figure \ref{fig:subtree_input}. These trees are rooted at vertices $v_a$ and $v_b$ respectively. We define input gates of $F$ to have a height of one, so the trees in Figure \ref{fig:subtree_input} satisfy the first invariant.
The remaining two invariants can be verified by examining every case of Figure \ref{fig:subtree_input}.

\bigskip
\noindent
\textbf{AND Gate.}
Given an input gate $g = (g_1 \land g_2)$, and the trees $T_a^{1}, T_b^{1}$ and $T_a^{2}, T_b^{2}$ corresponding to gates $g_1$ and $g_2$ respectively, we wish to construct  trees $T_a^g$ and $T_b^g$ so that $T_a^g$ is contained in $T_b^g$ if and only if gate $g$ has output $1$ on input $(a, b)$. By our third invariant it suffices to ensure that $T_a^g$ is contained in $T_b^g$ if and only if $T_a^1$ is contained in $T_b^1$ AND $T_a^2$ is contained in $T_b^2$. We construct trees $T_a^g$ and $T_b^g$ as in Figure \ref{fig:subtree_AND}. The trees are rooted at vertices $v_a^0$ and $v_b^0$ respectively. 
We now verify that all invariants are satisfied.

\begin{itemize}
    \item \textbf{Invariant 1.} 
    By our inductive hypothesis tree $T_a^1$ has the same height as $T_b^1$ and $T_a^2$ has the same height as $T_b^2$, so it follows from our construction that $T_a^g$ has the same height as $T_b^g$. Now to see why the height of these trees is at most $4h$, note that subtrees $T_a^1, T_b^1, T_a^2, T_b^2$ have height at most $4(h-1)$, and so trees $T_a^g$ and $T_b^g$ have height at most $4(h-1) + 4 = 4h$. 
    
    
    
    \item \textbf{Invariant 2.} We assume that the construction of trees $T_a^1$ and $T_a^2$ is independent of $b$, and the trees $T_b^1$ and $T_b^2$ are independent of $a$.  Then it can be easily verified that tree $T_a^g$ does not depend on $b$, and tree $T_b^g$ does not depend on $a$. 
    
    \item \textbf{Invariant 3.} We must  show that tree $T_a^g$ is contained in tree $T_b^g$ if and only if $g$ evaluates to $1$ on bit assignment $(a, b)$. 
    By our inductive hypothesis, it suffices to show that $T_a^g$ is contained in $T_b^g$ if and only if $T_a^1$ is contained in $T_b^1$ AND $T_a^2$ is contained in $T_b^2$.
    The `if' direction is immediate from our construction: just map vertex $v_a^i$ in $T_a^g$ to vertex $v_b^i$ in $T_a^g$ for $i \in [0, 4]$, and map trees $T_a^1$ and $T_b^1$ to subtrees of $T_a^2$ and $T_b^2$ respectively. 
    
    For the `only if' direction we must prove that subtree $T_a^1$ can only map to a subtree of $T_b^1$, and subtree $T_a^2$ can only map to a subtree of $T_b^2$. First note that since trees $T_a^g$ and $T_b^g$ have the same height, every isomorphism between $T_a^g$ and a subtree $T_b^g$ must map the root vertex $v_a^0$ of $T_a^g$ to the root vertex $v_b^0$ of $T_b^g$. Now suppose $T_a^1$ is mapped to $T_b^2$ in some isomorphism between $T_a^g$ and a subtree of $T_b^g$. Then  vertex $v_a^3$ would be mapped to vertex $v_b^4$, and the path of length two hanging off $v_a^3$ would have nowhere to map to. It immediately follows that in every valid subtree isomorphism,  $T_a^1$ is mapped to  $T_b^1$, and  $T_a^2$ is mapped to $T_b^2$. Then $T_a^g$ is contained in $T_b^g$ if and only if $T_a^1$ is contained in $T_b^1$ and $T_a^2$ is contained in $T_b^2$. 
    

    

\end{itemize}


\begin{figure}
\centering
\begin{minipage}{.35\textwidth}
\resizebox{1\textwidth}{!}{
\begin{tikzpicture}[
  lightnode/.style = {circle, fill = black, inner sep=2pt, solid},
  emptynode/.style={},
  level 1/.style = {level distance=10mm, sibling distance=1.2cm},
  level 2/.style = {level distance=10mm, sibling distance=1.2cm},
  level 3/.style = {level distance=10mm, sibling distance=1.2cm},
  level 4/.style = {level distance=10mm, sibling distance=1.2cm},
  triangle/.style={
  draw,solid,shape border uses incircle,
  isosceles triangle,shape border rotate=90,yshift=-.8cm, minimum width=2cm},
]

\node[lightnode, label=above:{\Large $v_a^0$}] {}
    child[line width=.5pt] {
            node[lightnode] (r1_1){}
    }
    child[line width=.5pt]{
        node[emptynode](r1_2){}
        child[line width=.5pt] {
            node[emptynode] (r3_5) {}
            edge from parent[draw = none]
        }
        child[line width=.5pt]{
            node[lightnode, label=above:{\Large $1$}] (r2_3) {}
            edge from parent[draw = none]
        }
        child[line width=.5pt] {
            node[lightnode, label=right:{\Large $v_a^1$}] (r2_4) {}
            edge from parent[draw=none]
        }
        child[line width=.5pt] {
            node[emptynode] (r2_5) {}
            child[line width=.5pt]{
                node[lightnode,label=above:{\Large $2$}] (r3_2) {}
                edge from parent[draw = none]
            }
            child[line width=.5pt]{
                node[lightnode] (r3_3) {}
                edge from parent[draw = none]
            }
            child[line width=.5pt] {
                node[lightnode, label=right:{\Large $v_a^3$}] (r3_4) {}
                child[line width=.5pt] {
                    node[lightnode] (r4_1) {}
                    child[line width=.5pt]{
                        node[triangle] {\Large $T_a^1$}
                    }
                }
                edge from parent[draw = none]
            }
            child[line width=.5pt] {
                node[emptynode] (r3_5) {}
                edge from parent[draw = none]
            }
            child[line width=.5pt] {
            node[lightnode, label=left:{\Large $v_a^4$}] (r3_6) {}
            child[line width=.5pt] {
                    node[lightnode] (r4_2) {}
                    child[line width=.5pt]{
                        node[triangle] {\Large $T_a^2$}
                    }
                }
            edge from parent[draw = none]
            }
            child[line width=.5pt]{
                node[lightnode,  label=above:{\Large $1$}] (r3_7) {}
                edge from parent[draw = none]
            }
            child[line width=.5pt] {
                node[emptynode] {}
                edge from parent[draw = none]
            }
            edge from parent[draw = none]
        }
        child[line width=.5pt] {
            node[lightnode, label=left:{\Large $v_a^2$}] (r2_6) {}
            edge from parent[draw=none]
        }
        child[line width=.5pt]{
            node[lightnode] (r2_7) {}
            edge from parent[draw = none]
        }
        child[line width=.5pt]{
            node[lightnode, label=above:{\Large $2$}] (r2_8) {}
            edge from parent[draw = none]
        }
        edge from parent[draw = none]
    }
    child[line width=.5pt] {
            node[lightnode](r1_3){}
    }
    ;

\draw[solid] (r1_1) to (r2_4);
\draw[solid] (r1_3) to (r2_6);

\draw[solid] (r2_3) to (r2_4);
\draw[solid] (r2_6) to (r2_7);
\draw[solid] (r2_7) to (r2_8);

\draw[solid] (r2_4) to (r3_4);
\draw[solid] (r2_6) to (r3_6);

\draw[solid] (r3_2) to (r3_3);
\draw[solid] (r3_3) to (r3_4);
\draw[solid] (r3_6) to (r3_7);
\end{tikzpicture}
}
\end{minipage}
\begin{minipage}{.6\textwidth}
\resizebox{\textwidth}{!}{
\begin{tikzpicture}[
  lightnode/.style = {circle, fill = black, inner sep=2pt, solid},
  emptynode/.style = {circle},
  level 1/.style = {level distance=10mm, sibling distance=3cm},
  level 2/.style = {level distance=10mm, sibling distance=1.3cm},
  level 3/.style = {level distance=10mm, sibling distance=1.3cm},
  level 4/.style = {level distance=10mm, sibling distance=1.2cm},
  triangle/.style={isosceles triangle,
  draw,solid,shape border uses incircle,
  shape border rotate=90,yshift=-.8cm, minimum width=2cm},
]

\node[lightnode, label=above:{\Large $v_b^0$}] {}
    child[line width=.5pt]{
        node[lightnode] (r1_1) {}
    }
    child[line width=.5pt]{
        node[emptynode] {}
        child[line width=.5pt]{
            node[lightnode, label=above:{$1$}] (r2_3) {}
            edge from parent[draw = none]
        }
        child[line width=.5pt] {
            node[lightnode, label=right:{\Large $v_b^1$}] (r2_4) {}
            child[line width=.5pt]{
                node[lightnode,label=above:{\Large $2$}] (r3_2) {}
                edge from parent[draw = none]
            }
            child[line width=.5pt]{
                node[lightnode,] (r3_3) {}
                edge from parent[draw = none]
            }
            child[line width=.5pt] {
                node[lightnode, label=right:{\Large $v_b^4$}] (r3_4) {}
                child[line width=.5pt]{
                    node[lightnode] (r4_1) {}
                    child[line width=.5pt]{
                        node[triangle] {\Large $T_b^1$}
                    }
                }
                edge from parent[draw = none]
            }
            child[line width=.5pt] {
                node[emptynode] (a) {}
                edge from parent[draw = none]
            }
            child[line width=.5pt] {
                node[emptynode] (b) {}
                edge from parent[draw = none]
            }
            edge from parent[draw = none]
        }
        child[line width=.5pt] {
            node[emptynode] (r2_5) {}
            child[line width=.5pt] {
                node[emptynode] (c) {}
                edge from parent[draw = none]
            }
            edge from parent[draw = none]
        }
        child[line width=.5pt] {
            node[lightnode, label=left:{\Large $v_b^2$}] (r2_6) {}
            child[line width=.5pt] {
                node[emptynode] (r3_5) {}
                edge from parent[draw = none]
            }
            child[line width=.5pt] {
                node[lightnode, label=left:{\Large $v_b^5$}] (r3_6) {}
                child[line width=.5pt]{
                    node[lightnode] (r4_2) {}
                    child[line width=.5pt]{
                        node[triangle] {\Large $T_b^2$}
                    }
                    edge from parent[solid]
                }
                edge from parent[draw = none]
            }
            child[line width=.5pt]{
                node[lightnode,  label=above:{\Large $1$}] (r3_7) {}
                edge from parent[draw = none]
            }
            edge from parent[draw = none]
        }
        child[line width=.5pt]{
            node[lightnode] (r2_7) {}
            edge from parent[draw = none]
        }
        child[line width=.5pt]{
            node[lightnode, label=above:{2}] (r2_8) {}
            edge from parent[draw = none]
        }
        child[line width=.5pt] {
            node[emptynode] (r2_10) {}
            edge from parent[draw = none]
        }
        child[line width=.5pt] {
            node[lightnode, label=left:{\Large $v_b^3$}] (r2_11) {}
            child[line width=.5pt] {
                node[emptynode] (r3_10) {}
                edge from parent[draw = none]
            }
            child[line width=.5pt] {
                node[emptynode] (r3_10) {}
                edge from parent[draw = none]
            }
            child[line width=.5pt] {
                node[lightnode, label=left:{\Large $v_b^6$}] (r3_11) {}
                child[line width=.5pt]{
                    node[lightnode] (r4_3) {}
                    child[line width=.5pt]{
                        node[triangle] {\Large $U_g$}
                    }
                    edge from parent[solid]
                }
                edge from parent[draw = none]
            }
            child[line width=.5pt] {
                node[lightnode] (r3_12) {}
                edge from parent[draw = none]
            }
            child[line width=.5pt] {
                node[lightnode,label=above:{\Large $2$}] (r3_13) {}
                edge from parent[draw = none]
            }
            edge from parent[draw = none]
        }
        child[line width=.5pt] {
            node[lightnode] (r2_12) {}
            edge from parent[draw = none]
        }
        child[line width=.5pt] {
            node[lightnode,label=above:{\Large $2$}] (r2_13) {}
            edge from parent[draw = none]
        }
        edge from parent[draw = none]
    }
    child[line width=.5pt]{
        node[lightnode] (r1_2) {}
    }
    ;

\draw[solid] (r1_1) to (r2_4);
\draw[solid] (r1_1) to (r2_6);
\draw[solid] (r1_2) to (r2_11);

\draw[solid] (r2_3) to (r2_4);

\draw[solid] (r2_6) to (r2_7);
\draw[solid] (r2_7) to (r2_8);

\draw[solid] (r2_11) to (r2_12);
\draw[solid] (r2_12) to (r2_13);

\draw[solid] (r2_4) to (r3_4);
\draw[solid] (r2_6) to (r3_6);
\draw[solid] (r2_11) to (r3_11);

\draw[solid] (r3_2) to (r3_3);
\draw[solid] (r3_3) to (r3_4);

\draw[solid] (r3_6) to (r3_7);

\draw[solid] (r3_11) to (r3_12);
\draw[solid] (r3_12) to (r3_13);

\end{tikzpicture}
}
\end{minipage}
\caption{The trees $T_a^g$ (left) and $T_b^g$ (right) corresponding to OR gate $g = (g_1 \lor g_2)$. }
\label{fig:OR_trees}
\end{figure}
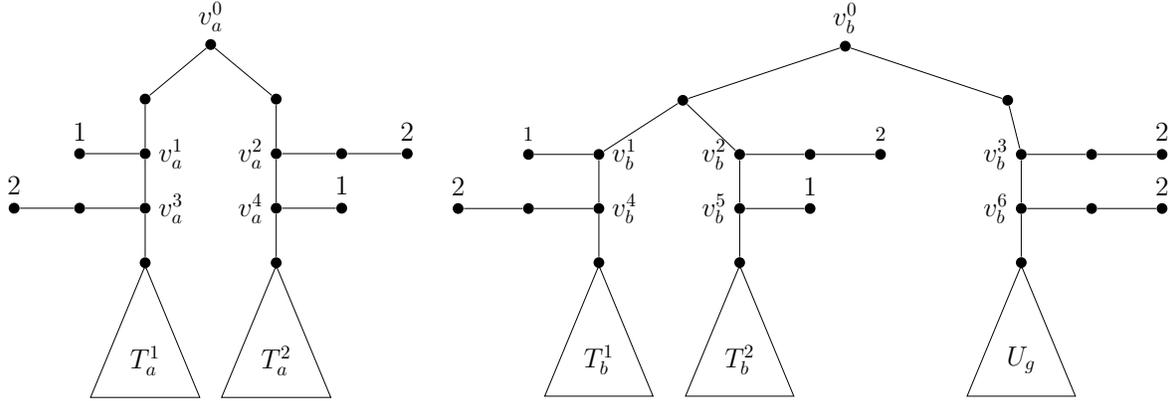

\bigskip
\noindent
\textbf{OR Gate.}
Given an input gate $g = (g_1 \lor g_2)$, and the trees $T_a^{1}, T_b^{1}$ and $T_a^{2}, T_b^{2}$ corresponding to gates $g_1$ and $g_2$ respectively, we will construct trees $T_a^g$ and $T_b^g$ so that $T_a^g$ is contained in $T_b^g$ if and only if $T_a^1$ is contained in $T_b^1$ OR $T_a^2$ is contained in $T_b^2$. We construct trees $T_a^g$ and $T_b^g$ as in Figure \ref{fig:OR_trees}. These trees are rooted at vertices $v_a^0$ and $v_b^0$ respectively. Tree $T_b^g$ contains a subtree $U_g$, which we call a universal subtree. We  design $U_g$ so that it contains both tree $T_a^1$ and tree $T_a^2$ for every bit assignment $a$. This will allow either $T_a^1$ or $T_a^2$ to match with $U_g$, thus achieving the OR gate logic. 

We now construct our universal subtree $U_g$. First, observe that for any gate $g$ and any two bit assignments $a, a' \in A$, the only difference between trees $T_a^g$ and $T_{a'}^g$ is in the input gate subtrees. There are two different input gate subtrees in $T_a^g$: the $a_i = 0$ subtree composed of a root vertex and two leaves, and the $a_i = 1$ subtree composed of a root vertex with a single leaf (see Figure \ref{fig:subtree_input}). Note that the $a_i = 0$ input subtree contains the $a_i = 1$ input subtree. Then if we define a bit assignment $u = {0}^m$, it follows that for every $a \in A$, the tree $T_a^g$ is contained within the tree $T_u^g$. Then for trees $T_a^1$ and $T_a^2$ we  construct trees $T_u^1$ and $T_u^2$ so that $T_a^1$ is contained in $T_u^1$ and $T_a^2$ is contained in $T_u^2$ for all $a \in A$. We define our universal subtree $U_g$ as the tree created by merging the root vertex of $T_u^1$ with the root vertex of $T_u^2$. By construction, this tree $U_g$ contains $T_a^1$ and $T_a^2$ for all $a \in A$ as intended. We now verify that all invariants are satisfied.

\begin{itemize}
    \item \textbf{Invariant 1.} This invariant holds by an argument identical to that of the AND gate construction. 
    
    \item \textbf{Invariant 2.} A similar argument as with the AND gate will show that $T_a^g$ does not depend on bit assignment $b$.  Likewise, tree $T_b^g$ does not depend on bit assignment $a$; the construction of universal subtree $U_g$ is independent of $a$ as detailed in its construction.
    
    \item \textbf{Invariant 3.} By our inductive hypothesis, it suffices to show that $T_a^g$ is contained in $T_b^g$ if and only if $T_a^1$ is contained in $T_b^1$ OR $T_a^2$ is contained in $T_b^2$. The `if' direction can be seen by observing that if $T_a^1$ is contained in $T_b^1$, then we can align $T_a^1$ with $T_b^1$ and align $T_a^2$ with $U_g$, which is guaranteed to contain $T_a^2$; the case where $T_a^2$ is contained in $T_b^2$ is identical. 
    
    The `only if' direction follows from a similar argument given for the AND construction. First note that since trees $T_a^g$ and $T_b^g$ have the same height, every subtree isomorphism must map the root vertex $v_a^0$ of $T_a^g$ to the root vertex $v_b^0$ of $T_b^g$. Additionally, it is immediate from construction that exactly one subtree $T_a^1$ or $T_a^2$ can be aligned with universal subtree $U_g$. Then we simply need to verify that there is no valid subtree isomorphism between $T_a^g$ and $T_b^g$ that maps $T_a^1$ to $T_b^2$ or $T_a^2$ to $T_b^1$. Suppose that  $T_a^1$ was mapped to a subtree of $T_b^2$ (the other case is symmetric). Then vertex $v_a^3$ would map to vertex $v_b^5$, and the path of length two hanging off $v_a^3$ would have nowhere to map to. We conclude that subtree $T_a^1$ must map to subtree $T_b^1$ or subtree $T_a^2$ must map to subtree $T_b^2$ in any subtree isomorphism from $T_a^g$ to $T_b^g$. The invariant is maintained.

    
\end{itemize}

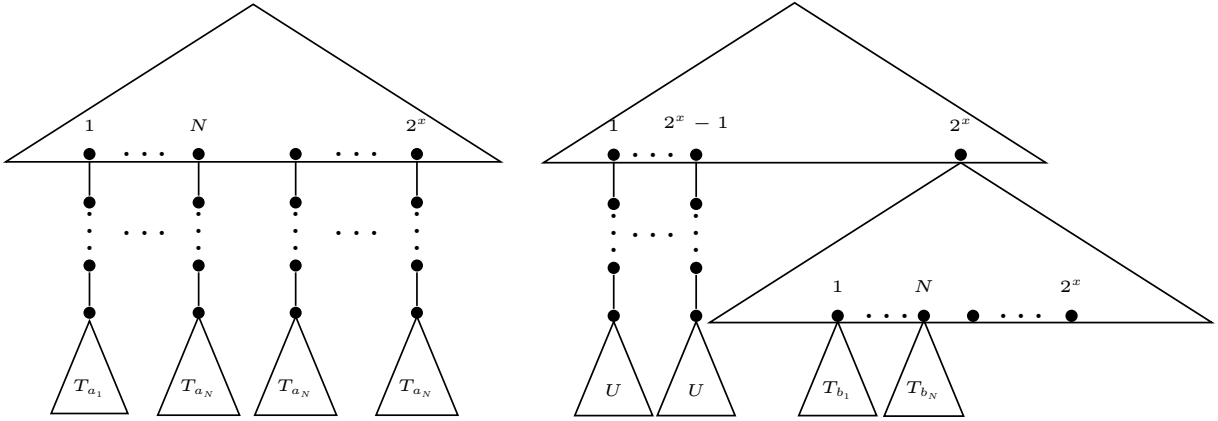
\begin{figure}
\centering
\begin{minipage}{.425\textwidth}
\resizebox{1\textwidth}{!}{
\begin{tikzpicture}[
  lightnode/.style = {circle, fill = black, inner sep=1pt, solid},
  emptynode/.style={},
  level 1/.style = {level distance=5.3mm, sibling distance=7mm},
  level 2/.style = {level distance=4mm, sibling distance=6mm},
  level 3/.style = {level distance=1.8mm, sibling distance=1mm},
  level 4/.style = {level distance=3.9mm, sibling distance=.1cm},
  triangle/.style={
  draw,solid,shape border uses incircle,
  isosceles triangle,
  isosceles triangle apex angle=115,
  shape border rotate=90,yshift=0cm, minimum width=2.5cm, minimum height=1.3cm, fill=none
  },
  smalltriangle/.style={
  draw,solid,shape border uses incircle,
  isosceles triangle,shape border rotate=90,yshift=-2.3mm, minimum width=18pt, fill=none,inner sep=.1pt},
]

\node[triangle] {}
    child[]{
        node[lightnode,label=above:{\fontsize{4}{6}\selectfont $1$},xshift=4mm] {}
        edge from parent[draw=none]
        child[]{
            node[lightnode] {}
            child[]{
                node[emptynode,yshift=.5mm] {\fontsize{9}{11}\selectfont$\vdots$}
                child[]{
                    node[lightnode] {}
                    child[]{
                        node[lightnode] {}
                        child[]{
                            node[smalltriangle] {\fontsize{4}{6}\selectfont $T_{a_1}$}
                            edge from parent[draw=none]
                        }
                    }
                    edge from parent[draw=none]
                }
                edge from parent[draw=none]
            }
        }
    }
    child[]{
        node[emptynode,xshift=1.5mm] {\fontsize{9}{11}\selectfont$\hdots$}
        edge from parent[draw=none]
        child[]{
            node[emptynode,yshift=-2.5mm] {\fontsize{9}{11}\selectfont$\hdots$}
            edge from parent[draw=none]
        }
    }
    child[]{
        node[lightnode,xshift=-1mm,label=above:{\fontsize{4}{6}\selectfont$N$}] {}
        child[]{
            node[lightnode] {}
            child[]{
                node[emptynode,yshift=.5mm] {\fontsize{9}{11}\selectfont$\vdots$}
                child[]{
                    node[lightnode] {}
                    child[]{
                        node[lightnode] {}
                        child[]{
                            node[smalltriangle] {\fontsize{4}{6}\selectfont $T_{a_N}$}
                            edge from parent[draw=none]
                        }
                    }
                    edge from parent[draw=none]
                }
                edge from parent[draw=none]
            }
        }
        edge from parent[draw=none]
    }
    child[]{
        node[lightnode,xshift=0mm] {}
        child[]{
            node[lightnode] {}
            child[]{
                node[emptynode,yshift=.5mm] {\fontsize{9}{11}\selectfont$\vdots$}
                child[]{
                    node[lightnode] {}
                    child[]{
                        node[lightnode] {}
                        child[]{
                            node[smalltriangle] {\fontsize{4}{6}\selectfont $T_{a_N}$}
                            edge from parent[draw=none]
                        }
                    }
                    edge from parent[draw=none]
                }
                edge from parent[draw=none]
            }
        }
        edge from parent[draw=none]
    }
    child[]{
        node[emptynode,xshift=-2mm] {\fontsize{9}{11}\selectfont$\hdots$}
        edge from parent[draw=none]
        child[]{
            node[emptynode,yshift=-2.5mm] {\fontsize{9}{11}\selectfont$\hdots$}
            edge from parent[draw=none]
        }
    }
    child[]{
        node[lightnode,xshift=-4mm,label=above:{\fontsize{4}{6}\selectfont $2^x$}] {}
        child[]{
            node[lightnode] {}
            child[]{
                node[emptynode,yshift=.5mm] {\fontsize{9}{11}\selectfont$\vdots$}
                child[]{
                    node[lightnode] {}
                    child[]{
                        node[lightnode] {}
                        child[]{
                            node[smalltriangle] {\fontsize{4}{6}\selectfont $T_{a_N}$}
                            edge from parent[draw=none]
                        }
                    }
                    edge from parent[draw=none]
                }
                edge from parent[draw=none]
            }
        }
        edge from parent[draw=none]
    }
;

\end{tikzpicture}
}
\end{minipage}
\begin{minipage}{.565\textwidth}
\resizebox{\textwidth}{!}{
\begin{tikzpicture}[
  lightnode/.style = {circle, fill = black, inner sep=1pt, solid},
  emptynode/.style={},
  level 1/.style = {level distance=5.3mm, sibling distance=.6cm},
  level 2/.style = {level distance=4mm, sibling distance=3mm},
  level 3/.style = {level distance=1.8mm, sibling distance=.1mm},
  level 4/.style = {level distance=3.9mm, sibling distance=.1cm},
  triangle/.style={
  draw,solid,shape border uses incircle,
  isosceles triangle,
  isosceles triangle apex angle=115,
  shape border rotate=90,yshift=0cm, minimum width=2.5cm, minimum height=1.3cm, fill=none
  },
  smalltriangle/.style={
  draw,solid,shape border uses incircle,
  isosceles triangle,shape border rotate=90,yshift=-2.1mm, minimum width=18pt, fill=none,inner sep=.1pt},
]

\node[triangle] {}
    child[]{
        node[lightnode,label=above:{\fontsize{4}{6}\selectfont $1$},xshift=-5.7mm] {}
        child[]{
            node[lightnode] {}
            child[]{
                node[emptynode,yshift=.5mm] {\fontsize{9}{11}\selectfont$\vdots$}
                child[]{
                    node[lightnode] {}
                    child[]{
                        node[lightnode] {}
                        child[]{
                            node[smalltriangle] {\fontsize{4}{6}\selectfont $U$}
                            edge from parent[draw=none]
                        }
                    }
                    edge from parent[draw=none]
                }
                edge from parent[draw=none]
            }
        }
        edge from parent[draw=none]
    }
    child[]{
        node[emptynode,xshift=-8.5mm] {\fontsize{9}{11}\selectfont$\hdots$}
        edge from parent[draw=none]
        child[]{
            node[emptynode,yshift=-2.5mm] {\fontsize{9}{11}\selectfont$\hdots$}
            edge from parent[draw=none]
        }
    }
    child[]{
        node[lightnode,xshift=-11mm,label=above:{\fontsize{4}{6}\selectfont $2^x-1$}] {}
        child[]{
            node[lightnode] {}
            child[]{
                node[emptynode,yshift=.5mm] {\fontsize{9}{11}\selectfont$\vdots$}
                child[]{
                    node[lightnode] {}
                    child[]{
                        node[lightnode] {}
                        child[]{
                            node[smalltriangle] {\fontsize{4}{6}\selectfont $U$}
                            edge from parent[draw=none]
                        }
                    }
                    edge from parent[draw=none]
                }
                edge from parent[draw=none]
            }
        }
        edge from parent[draw=none]
    }
    child[]{
        node[lightnode,label=above:{\fontsize{4}{6}\selectfont $2^x$},xshift=4.5mm] {}
        edge from parent[draw=none]
        child[]{
            node[emptynode]{}
            child[]{
                node[emptynode]{}
                child[]{
                    node[emptynode,yshift=.5mm]{}
                    child[]{
                        node[lightnode,xshift=-1mm,label=above:{\fontsize{4}{6}\selectfont $1$}] {}
                        edge from parent[draw=none]
                        child[]{
                            node[smalltriangle] {\fontsize{4}{6}\selectfont $T_{b_1}$}
                            edge from parent[draw=none]
                        }
                    }
                    edge from parent[draw=none]
                }
                edge from parent[draw=none]
            }
            edge from parent[draw=none]
        }
        child[]{
            node[emptynode]{}
            child{
                node[emptynode]{}
                child{
                    node[emptynode,yshift=.5mm]{}
                    child{
                        node[emptynode,xshift=0mm] {\fontsize{9}{11}\selectfont$\hdots$}
                        edge from parent[draw=none]
                    }
                    edge from parent[draw=none]
                }
                edge from parent[draw=none]
            }
            edge from parent[draw=none]
        }
        child[]{
            node[emptynode]{}
            child{
                node[emptynode]{}
                child{
                    node[emptynode,yshift=.5mm]{}
                    child{
                        node[lightnode,xshift=0mm,label=above:{\fontsize{4}{6}\selectfont $N$}] {}
                        edge from parent[draw=none]
                        child[]{
                            node[smalltriangle] {\fontsize{4}{6}\selectfont $T_{b_N}$}
                            edge from parent[draw=none]
                        }
                    }
                    edge from parent[draw=none]
                }
                edge from parent[draw=none]
            }
            edge from parent[draw=none]
        }
        child[]{
            node[triangle,yshift=-3.7mm] {}
            edge from parent[draw=none]
        }
        child[]{
            node[emptynode]{}
            child{
                node[emptynode]{}
                child{
                    node[emptynode,yshift=.5mm]{}
                    child{
                        node[lightnode,xshift=-2mm] {}
                        edge from parent[draw=none]
                    }
                    edge from parent[draw=none]
                }
                edge from parent[draw=none]
            }
            edge from parent[draw=none]
        }
        child[]{
            node[emptynode]{}
            child{
                node[emptynode]{}
                child{
                    node[emptynode,yshift=.5mm]{}
                    child{
                        node[emptynode,xshift=-1.1mm] {\fontsize{9}{11}\selectfont$\hdots$}
                        edge from parent[draw=none]
                    }
                    edge from parent[draw=none]
                }
                edge from parent[draw=none]
            }
            edge from parent[draw=none]
        }
        child[]{
            node[emptynode] {}
            child[]{
                node[emptynode] {}
                child[]{
                node[emptynode,yshift=.5mm] {}
                    child[]{
                        node[lightnode,label=above:{\fontsize{4}{6}\selectfont $2^x$}] {}
                        edge from parent[draw=none]
                    }
                    edge from parent[draw=none]
                }
                edge from parent[draw=none]
            }
            edge from parent[draw=none]
        }
    }
;

\end{tikzpicture}
}
\end{minipage}
\caption{The final $T_A$ (left) and $T_B$ (right). }
\label{fig:final}
\end{figure}

\subsection{Completing the Reduction}
The final trees are constructed using the technique provided in \cite{DBLP:journals/talg/AbboudBHWZ18}. The construction is shown in Figure \ref{fig:final} and described next.
\begin{itemize}
    \item For the final tree $T_A$, start with a complete binary tree where the number of leaves is the smallest power of $2$ that is greater or equal to $N$, say $2^x$. From each of the $2^x$ leaves, attach a path of length $x$. Let the first $N$ leaves at the ends of these paths be numbered $1$ to $N$. For $1 \leq i \leq N$, replace leaf $i$ with root of $T_{a_i}$. For the remaining $2^x - N$ leaves at the end of paths, replace the leaf with the roots of $2^x - N$ copies of $T_{a_N}$.
    
    \item For the final tree $T_B$, again start with a complete binary tree with $2^x$ leaves. From the first $2^x -1$ leaves, attach a path of length $x$. Replace the end of each of the paths with the root of a universal tree $U$, which is $T_a$ with input bit assignment $u = 0^m$. From the remaining leaf in the complete binary tree, replace this leaf with the root of another complete binary tree, again with $2^x$ leaves. Let the first $N$ leaves of this second complete binary tree be numbered $1$ to $N$. For $1 \leq i \leq N$, replace leaf $i$ with the root of $T_{b_i}$. 
\end{itemize}
To see why this works, consider that for $T_A$ to be isomorphic to a subtree of $T_B$, the root of $T_A$ must be mapped onto the root of $T_B$. Then, one of $T_A$'s $2^x$ paths hanging from the leaves of its complete binary tree must traverse down the lower complete binary tree in $T_B$. From here, a subtree rooted at the end of one of these paths in $T_A$ must have to be isomorphic to one of the subtrees hanging from the leaves of the second binary tree in $T_B$. This is possible if and only if for some $a \in A$ and $b \in B$ we have that $T_{a}$ is isomorphic to a subtree of $T_{b}$. By the invariants proven above, such a pair $a \in A$ and $b \in B$ exists iff the starting formula $F$ evaluates to true on the assignment $(a, b)$.

The final tree $T_A$ is of size $\mathcal{O}(Ns)$. This is because there are $N$ trees $T_a$ in $T_A$, and each tree $T_a$ is of size $\mathcal{O}(s)$. The upper bound on the size of $T_a$ follows from the fact that formula $F$ has $s$ gates, and each gate contributes constantly many vertices to $T_a$.
The final tree $T_B$ is of size $\mathcal{O}(N s^2)$.  To see this, fix a particular assignment $(a, b)$, and consider the tree $T_b$. Each AND gate contributes a constant number of vertices to $T_b$. Each OR gate appends a universal subtree $U$ of size at most the size of $T_a$ to $T_b$. Since the size of $T_a$ is $\mathcal{O}(s)$ and there are $s$ gates in formula $F$, we have that $T_b$ is of size $\mathcal{O}(s^2)$.

\section{Discussion}
\label{sec:discussion}
The key property highlighted by the two reductions is that both problems we reduced to allow for the construction of two independent objects $O_A$ and $O_B$, where $O_A$ is constructed independently from the partial input assignments in $B$, and $O_B$ is constructed independently from the partial input assignments in $A$. 

In order to construct these objects, both reductions start by fixing an input assignment $(a,b)$. Then, two new objects for each gate $g$ are constructed using the objects for the circuits that are input into $g$. The aim of this construction is to maintain the invariant that whichever desired property we want our objects to have (e.g., the pattern occurring in a graph, or having an isomorphic subtree) holds iff $(a,b)$ satisfy the circuit with output gate $g$. This is accomplished by supposing (i) we are adding the gate $g = g_1 \ast g_2$ where $\ast \in \{\land, \lor\}$, (ii) the objects $O_a^{g_1}$ and $O_b^{g_1}$ have the desired property iff $(a,b)$ evaluates to true on the circuit with output gate $g_1$, and (iii) the objects $O_a^{g_2}$ and $O_b^{g_2}$ have the desired property iff $(a,b)$ evaluate to true on the circuit with output gate $g_2$. The task is then to construct $O_a^g$ from only $O_a^{g_1}$ and $O_a^{g_2}$, and $O_b^g$ from only $O_b^{g_1}$ and $O_b^{g_2}$, such that $O_a^g$ and $O_b^g$ have the desired property iff $g = g_1 \ast g_2$ evaluates to true. By the invariant, this is equivalent when $\ast = \land$ to $O_a^{g_1}$ and $O_b^{g_1}$ having the desired property, and $O_a^{g_2}$ and $O_b^{g_2}$ having the desired property. In the case of $\ast = \lor$, only one of the pairs $O_a^{g_1}$, $O_b^{g_1}$ or $O_a^{g_2}$, $O_b^{g_2}$ needs to have the property.

In the last step, the final objects $O_A$ and $O_B$ are constructed by combining all $O_{a_i}$, $1\leq i \leq N$ to form $O_A$, and $O_{b_j}$, $1 \leq j \leq N$ to form $O_B$. These final objects must allow for selection between different partial assignments. Additionally, the final objects satisfy the desired property iff at least one object pair $O_{a_i}$ and $O_{b_j}$ together satisfy the desired property.

The above outlines, on a high level, the approach used in reductions from Formula-SAT to polynomial-time problems that appear here, and in \cite{DBLP:conf/icalp/AbboudB18,DBLP:journals/corr/abs-2008-02769}. The techniques presented in \cite{DBLP:conf/stoc/AbboudHWW16} instead start with the problem of the satisfiability of branching programs, but they work similarly in the sense that they must model the logical gates AND and OR (this time connecting logical statements about reachability). The authors also take similar steps in order to build two independent objects based on a fixed input assignment $(a,b)$.

\bibliography{ref}

\appendix
\section{Proving the implications of logarithmically faster algorithms for Subtree Isomorphism} \label{app:circuit}

\begin{thm}[\cite{DBLP:conf/stoc/AbboudHWW16}]
\label{thm:circuit_1}
Let $n \leq S(n) \leq 2^{o(n)}$ be time constructible and monotone non-decreasing.
Let $\mathcal{C}$ be a class of circuits.
Suppose there is an SAT algorithm for $n$-input circuits which are $ANDs$ of $\mathcal{O}(S(n))$
arbitrary functions of three $\mathcal{O}(S(n))$-size circuits from C, that runs in $\mathcal{O}(2^n/n^{10})$ time. Then $\E^\NP$ does not have $S(n)$-size circuits.
\end{thm}

\begin{thm}[\cite{DBLP:conf/stoc/AbboudHWW16}]
\label{thm:circuit_2}
Suppose there is a satisfiability algorithm for bounded fan-in formulas of size
$n^k$
running in $\mathcal{O}(2^n/n^k)$ time, for all constants $k > 0$. Then $\NTIME[2^{\mathcal{O}(n)}]$ is not contained in
non-uniform $\NC^1$.
\end{thm}

\noindent
\textbf{Corollary \ref{cor:strict_subquad_cons}.} 
\emph{The existence of a strongly subquadratic time algorithm for PMLG (or Subtree Isomorphism) would imply the class $\E^\NP$ (1) does not have non-uniform $2^{o(n)}$-size Boolean formulas and (2) does not have non-uniform $o(n)$-depth circuits of bounded fan-in. It also implies that $\NTIME[2^{\mathcal{O}(n)}]$ is not in non-uniform $\NC$.}

\begin{proof}
Note that the condition in Theorem \ref{thm:circuit_1} that the SAT-algorithm works on $n$-input circuits which are ANDs of $\mathcal{O}(S(n))$ arbitrary functions of three $\mathcal{O}(S(n))$-size circuits is trivially satisfied by a solver that works over Boolean formula. 
By Theorem \ref{thm_pmlg} (Theorem \ref{thm_sub_tree} resp.), for circuits (or equivalently formulas) of size $S(n) = 2^{o(n)}$, a strongly subquadric time algorithm for PMLG (Subtree Isomorphism resp.) would imply a SAT algorithm running in time $$\mathcal{O}(n^{1 + o(1)}\cdot|E||P|^{1-\varepsilon}) = \mathcal{O}(n^{1 + o(1)}\cdot2^{n - \varepsilon n/2}S(n)^{4})$$ which is $\mathcal{O}(2^{n}/n^{10})$; the $n^{1 + o(1)}$ factor is introduced when moving from a word size of $\Theta(\log n)$ to $\Theta(n)$. Thus, Theorem \ref{thm:circuit_1} implies (1).
Part (2) is implied as well since a $o(n)$-depth circuit of bounded fan-in can be expressed as a formula of size $S(n) = 2^{o(n)}$.
The last statement follows from Theorem \ref{thm:circuit_2} and the fact that on circuits of size $n^k$, our subquadratic algorithm would run in time $\mathcal{O}(n^{1 + o(1)}\cdot 2^{n - \varepsilon n/2}n^{2k})$ which is $\mathcal{O}(2^n / n^k)$.
\end{proof}

\noindent
\textbf{Corollary \ref{cor:infinite_logshaving_cons}.} 
\emph{If PMLG (or Subtree Isomorphism) can be solved in time $\mathcal{O}(\frac{|E||P|}{\log^c |E|})$ or $\mathcal{O}(\frac{|E||P|}{\log^c |P|})$ ( $\mathcal{O}(\frac{|T_1||T_2|}{\log^c |T_1|})$ or $\mathcal{O}(\frac{|T_1||T_2|}{\log^c |T_2|})$ resp.) for all $c = \Theta(1)$, then $\NTIME[2^{O(n)}]$ does not have non-uniform polynomial-size log-depth circuits.}

\begin{proof}
We prove this for PMLG, the proof for Subtree Isomorphism is similar. By Theorem \ref{thm:circuit_2}, it suffices to show that for all $k$, 
there exists an algorithm to check satisfiability of all bounded fan-in formulas of size $n^k$ running in time $O(2^{n}/n^k)$. 
Suppose that for all $c = \Theta(1)$, there exists an algorithm running in time $O(\frac{|E||P|}{\log^c |P|})$ or $O(\frac{|E||P|}{\log^c |E|})$. Then by Theorem \ref{thm_pmlg}, if we let $c > 4k+1$ we obtain an algorithm running in time 
$$
\frac{n^{1 + o(1)}\cdot 2^n s^{3}}{\log^{c}(2^\frac{n}{2} s^{2})} = \frac{n^{1 + o(1)}\cdot  2^n n^{3k}}{\log^{c} (2^\frac{n}{2} n^{2k})} \leq \frac{n^{1 + o(1)}\cdot 2^{n}n^{3k}}{\left(\frac{n}{2}\right)^{c}} = \frac{2^{n + c}}{n^{c-3k - 1 - o(1)}} = O\left(\frac{2^n}{n^k}\right)
$$
\end{proof}

\noindent
\textbf{Corollary \ref{corr:E_NP_log_shaving_constant}.} 
\emph{$\E^\NP$ cannot be computed by non-uniform formulas of cubic size if PMLG (or Subtree Isomorphism) can be solved in time $\mathcal{O}\left(\frac{|E| \cdot |P|}{\log ^ {20 + \varepsilon} |E|}\right)$ or $\mathcal{O}\left(\frac{|E| \cdot |P|}{\log ^ {20 + \varepsilon} |P|}\right)$ for $\varepsilon > 0$, where $G$ is a deterministic DAG of maximum degree three (or $\mathcal{O}\left(\frac{|T_1| \cdot |T_2|}{\log^{20 + \varepsilon} |T_1|}\right)$ or $\mathcal{O}\left(\frac{|T_1| \cdot |T_2|}{\log ^ {20 + \varepsilon} |T_2|}\right)$ for $\varepsilon > 0$ resp.).
}
\begin{proof}
Theorem \ref{thm:circuit_1} as given in \cite{DBLP:conf/stoc/AbboudHWW16} says that solving Formula-SAT in time $\mathcal{O}(2^n / n^{10})$ on formulas of size $s = \mathcal{O}(n^{3 + \varepsilon})$ implies that there is a function in class $\E^\NP$ that cannot be computed by formulas of size $\mathcal{O}(n^{3 + \varepsilon})$. Then via a proof identical to that of Corollary \ref{cor:pmlg}, we have the above result.
\end{proof}

\end{document}